\documentclass[10pt,aps,pra,amssymb,superscriptaddress,twocolumn]{revtex4-2}
\usepackage{physics}
\usepackage{xcolor}
\usepackage{amssymb}
\usepackage{mathtools}
\usepackage{bbm}
\usepackage{graphicx} 
\usepackage{amsmath}
\usepackage{amsfonts}
\usepackage{braket}
\usepackage{fontenc}
\usepackage{mathrsfs}
\usepackage{verbatim}
\usepackage{dsfont}
\usepackage{csquotes}
\usepackage{comment}
\usepackage{amsthm}
\usepackage{xcolor}
\usepackage[colorlinks, citecolor=blue,urlcolor=blue, bookmarks=false,hypertexnames=true]{hyperref}
\usepackage{ulem}
\usepackage{bm}

\theoremstyle{plain}
\newtheorem{thm}{Theorem}[]

\newtheorem*{thm*}{Theorem}
\newtheorem{lem}[thm]{Lemma}
\newtheorem{prop}[thm]{Proposition}
\newtheorem{cor}[thm]{Corollary}

\theoremstyle{remark}

\newcommand{\sumab}[2]{\underset{#1}{\overset{#2}{\sum}}}
\newcommand{\suma}[1]{\underset{#1}{\sum}}

\begin{document} 

\title{Exact requirements for battery-assisted qubit gates}

\date{\today}

\author{Riccardo Castellano} \email{riccardo.castellano@unige.ch}
\affiliation{Scuola Normale Superiore, 56126 Pisa, Italy}
\affiliation{Dipartimento di Fisica dell’Universit\`a di Pisa, Largo Pontecorvo 3, I-56127 Pisa, Italy}
\affiliation{Department of Applied Physics, University of Geneva, 1211 Geneva, Switzerland}

\author{Vasco Cavina}
\email{vasco.cavina@sns.it}
\affiliation{Scuola Normale Superiore, 56126 Pisa, Italy}

\author{Martí Perarnau-Llobet}
\affiliation{F\'isica Te\`orica: Informaci\'o i Fen\`omens Qu\`antics, Department de F\'isica, Universitat Aut\`onoma de Barcelona, 08193 Bellaterra (Barcelona), Spain}
\affiliation{Department of Applied Physics, University of Geneva, 1211 Geneva, Switzerland}

\author{Pavel Sekatski}
\affiliation{Department of Applied Physics, University of Geneva, 1211 Geneva, Switzerland}

\author{Vittorio Giovannetti}
\affiliation{Scuola Normale Superiore, NEST, and Istituto Nanoscienze-CNR, 56126 Pisa, Italy}

\begin{abstract}

We consider the implementation of a unitary gate on a qubit system $S$ via a global energy-preserving operation acting on $S$ and an auxiliary system $B$ that can be seen as a battery. We derive a simple, asymptotically exact expression for the implementation error as a function of the battery state, which we refer to as the {\it Unitary Defect}. Remarkably, this quantity is independent of the specific gate being implemented, highlighting a universal property of the battery itself.  We show that minimizing the unitary defect, under given physical constraints on the battery state, is mathematically equivalent to solving a Lagrangian optimization problem, often corresponding to finding the ground state of a one-dimensional quantum system.
Using this mapping, we identify optimal battery states that achieve the highest precision under constraints on energy, squared energy, number of levels and Quantum Fisher Information. Overall, our results provide an efficient method  
for establishing  bounds on the physical requirements needed to implement a unitary gate via energy-preserving operations and for determining the corresponding optimal protocols.

\end{abstract}

\maketitle

The development of quantum information theory has prompted a renewed examination of symmetries and conservation laws, introducing an operational perspective
in which they are viewed as constraints that limit our ability to control and measure quantum systems~\cite{ReferenceFrames,
QResourceTh,Lipka-Bartosik2024}.
This concept was recognized as early as the 1950s, notably through the Wigner–Araki–Yanase (WAY) theorem \cite{(W)AYTh, W(A)YTh, WA(Y)Th}, which establishes fundamental limits on the realization of projective measurements in the presence of conserved quantities. Subsequent studies have shown that, in such scenarios, the measurement apparatus must be in a strongly asymmetric state in order to perform the measurement effectively \cite{WAY2006, ReferenceFrames, Ahmadi2013, EnergyMesuramentNavascues, WAY2008, WAYUnbounded(Tajima), WAY2011}.
Beyond measurements, conservation laws constrain our ability to manipulate quantum systems, affecting tasks ranging from state preparation~\cite{Marvian_2013StateCenversions, IIDNonAbelian, IID-Rates}, to unitary operations~\cite{VarRequest(2002), VarRequest(2003), VarRequest(general2009), LowerBoundNotgate(2007), Aberg2014, Tajima2018}, general quantum channels~\cite{ChannelsCost(Tajima)}, and work extraction in quantum thermodynamics~\cite{Lostaglio2015}.

A task of crucial importance, particularly in quantum computation, is the ability to perform unitary operations with a high precision.  
 Under the assumption that energy conservation is a fundamental symmetry, the implementation of non-energy-preserving gates (NEPGs) on a system 
$S$ requires the assistance of an auxiliary system $B$ - commonly referred to as a battery in quantum thermodynamics.
Following seminal papers on the topic \cite{VarRequest(2002), VarRequest(2003), VarRequest(general2009), LowerBoundNotgate(2007)} significant progress has been made in identifying necessary conditions that the battery must fulfil to achieve a given level of precision, i.e. generate an evolution of $S$ that approximates the desired gate sufficiently well. In particular, Quantum Fisher Information (QFI), average energy and coherent entropy of energy have been recognized as essential resources in this context~\cite{Chiribella, Chiribella2, WeakQFIbound, QFIBound, SaturatingWAY(Gaussian),EntropicCoherence(2025)}.
Nevertheless, a universal criterion to understand if a given state of the battery is sufficient to power the implementation of a given NEPG has not yet been found. Likewise, the form of the optimal state of the battery 
remains unknown. The goal of this paper is to address these open questions. 
In the first part of the paper we find a direct expression for the error with which battery states implement any gate $V_{S}$ in a two-level system, and show that is proportional to a new quantity, which we call {\it Unitary Defect} (UD).
Then, we explicitly identify the states that minimize the {\it Unitary Defect} under physically meaningful constraints on the initial state of the battery, namely, when it is prepared with a limited amount of energy or coherence.
 Finally, using these results, we formulate a new set of necessary conditions expressed through refined inequalities in which the gate error is bounded in terms of the battery’s initial resources. These conditions are shown to be stricter than those previously established in the literature. From a practical perspective, they show that semi-classical pulses are inefficient in terms of the error–energy trade-off, suggesting that the use of quantum batteries could enhance this aspect of quantum information processing.
 Our methods can be applied also to $d$-level systems with equally spaced energy levels, as discussed in App. F of \cite{SM}. 
 
 \begin{figure}
    \includegraphics[width=0.9\linewidth]{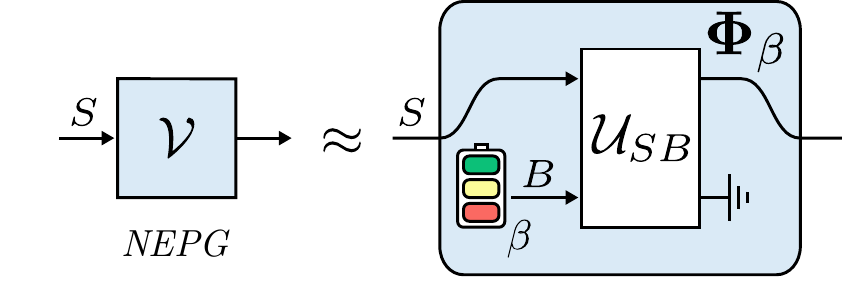}
    \caption{Schematic representation of the  approximated implementation of a non-energy-preserving gate (NEPG) $V_S$ on the system $S$ through  
    a joint 
    total-energy-preserving unitary transformation     $U_{SB}$ with the battery $B$ initialized in the input state $\beta_B$. }
    \label{fig:1}
\end{figure}
{\it  Mathematical formulation:--} Given a system $S$ with Hamiltonian ${H}_S$, our goal is to implement a generic   unitary gate ${V}_S$ using energy-preserving operations. 
If we treat $S$ as an isolated unit, this requirement restricts us to energy-preserving gates such that $[{V}_S, {H}_S] = 0$. To overcome this issue, we adopt a perspective where energy is conserved globally within a composite system consisting of $S$ and a battery $B$, with Hamiltonian ${H}_B$.
In this framework the most general evolution (or completely positive, trace preserving (CPTP) map) of $S$ takes the form~\cite{OpenQ.SystemsPetruccione} 
\begin{align} \label{Def:MainChannel}
      \mathbf{\Phi}_{\beta}(\cdots): =\mbox{Tr}_B[ {U}_{SB}(\cdots\otimes {\beta}_B){U}^\dag_{SB}], 
\end{align}
with $\beta_B$ the initial state of the battery, 
and with $U_{SB}$ a global unitary transformation
which commutes with the total Hamiltonian of the model, i.e. $[{U}_{SB},{H}_{S} + {H}_B] =0$.
Since we aim to implement a unitary gate $V_S$, we look for choices of ${\beta}$ and ${{U}_{SB}}$ granting that
$ \mathbf{\Phi}_{\beta}  \approx \mathcal{V}:={V}_S \cdots {V}_S^{\dag}$ (see Fig.~\ref{Def:MainChannel}). This requirement can be put on formal grounds by demanding the Choi infidelity $\epsilon_C(\mathbf{\Phi}_{\beta} ,\mathcal{V})$, a dissimilarity measure for quantum channels, to be small. 
The latter can be written in a compact form in terms of the Kraus operators~$\{K_S^{(k)}\}_k$ \cite{OpenQ.SystemsPetruccione} of the map~$\Phi_{\beta}$ \cite{ChoiFidGates}
 \begin{equation} \label{eq:choinf}
     \epsilon_{C}(\mathbf{\Phi}_{\beta},\mathcal{V})=1-\frac{1}{d^{2}}\underset{k}{\sum} \abs{\mbox{Tr}[V_S^{\dag}{K_S^{(k)}} ]}^{2}.
 \end{equation}     
Note that the precision in \cite{Chiribella, Chiribella2, WeakQFIbound, QFIBound} is given in terms of another function, the worst-case infidelity. However,  for finite-dimensional systems one can show (see App. A in~\cite{SM}
) that the two are equivalent in the following sense
 \begin{equation} \label{wcVSChoiFid}
\frac{\epsilon_{wc}(\mathbf{\Phi}_{\beta},\mathcal{V})}{d}\leq \epsilon_{C}(\mathbf{\Phi}_{\beta},\mathcal{V})\leq \epsilon_{wc}(\mathbf{\Phi}_{\beta},\mathcal{V}),
\end{equation}
where $d$ is the dimension of the system. To the best of our knowledge, the first bound in Eq. \eqref{wcVSChoiFid}, saturable up to a factor of 2, is a new result of independent interest.

{\it Calculation of the Choi infidelity:--} Using a similar approach to \cite{CataliticCoherence} we will assume that the system is a qubit with gap $\omega$ and the battery is a harmonic oscillator of the same frequency. 
This resonant condition is essential to ``activate" all the levels of the battery, i.e. make ${U}_{SB}$ exchange energy between $S$ and every couple of adjacent levels of $B$.
Summarizing, setting $\hbar =1$, we have
\begin{align} \label{eq:hams}
     &{H}_{S}=\frac{\omega}{2} (\underset{S}{\ketbra{1}{1}} - \underset{S}{\ketbra{0}{0}}),   \;\;  {H}_{B}= \omega \underset{n=0}{\overset{\infty}{\sum}} n\underset{B}{\ketbra{n}{n}},
    \end{align}
    with $|0\rangle_S$ and $|1\rangle_S$ the ground and excited state of $S$ and 
 $\ket{n}_B$  the $n$-th energy eigenvector of $B$.
    Since $\epsilon_{C}(\mathbf{\Phi}_{\beta},\mathcal{V})$ is linear in $\beta_B$ we assume the latter to be pure and denote with $\beta_n$ its amplitude on the $n$-th energy eigenvector. 
Finally, for a given target unitary operation ${V}_S$ we fix the global system-battery unitary in 
Eq.~\eqref{Def:MainChannel} as in \cite{Ergotropy, EnergyMesuramentNavascues, Aberg2014, Chiribella} to be
\begin{align} \notag
{U}_{SB} & = \underset{S}{\ketbra{0}{0}} \otimes \underset{B}{\ketbra{0}{0}}
 +\underset{n>0}{\overset{\infty}{\bigoplus}}\; {U}^{(n)}_{SB}, \\   {U}^{(n)}_{SB}  &:= \sum_{i,j=0,1} \underset{B}{\ketbra{n-i}{n-j}} \otimes \underset{S}{\ketbra{i}{j}}V_{ij}, \label{eq:usbUP}
\end{align}
where $V_{ij}:={_S\!}\bra{i} {V}_{S}\ket{j}_{S}$,
The choice of $U_{SB}$ is dictated by the energy preservation requirement: in the ground state of the $SB$ system the dynamics is necessarily trivial, while the evolution acts block-diagonally on each two-dimensional degenerate energy subspace of $SB$. Within these subspaces, we choose to \lq\lq copy" the target gate $V_{S}$ - we prove the optimality of this choice for real positive amplitude $\beta_n$ (as complex phases of any pure battery can be cancelled by an energy-preserving unitary) 
in the $\epsilon_{C}\rightarrow 0$ limit in appendix C of~\cite{SM}.  
From Eqs. \eqref{eq:choinf}, \eqref{eq:usbUP} we obtain an exact expression for the Choi infidelity
\begin{align} \notag
& \epsilon_C(\mathbf{\Phi}_{\beta},\mathcal{V}) = |V_{01}|^2 \sum_{n=1}^{\infty} \big[  |\beta_{n+1}  - \beta_n|^2  \\ & - \frac{|V_{01}|^2}{4}   |\beta_{n+1} + \beta_{n-1} - 2 \beta_n|^2 \big] + \Delta(\beta_{0},\beta_{1},V_{S}) \label{eq:finalchoi}
\end{align}
where $\Delta(\beta_{0},\beta_{1},V_{S})$ is a  \lq \lq boundary" term. The derivation of Eq. \eqref{eq:finalchoi} is done in App. C of~\cite{SM}.
In the following, we will assume $\beta_0 = 0$, which is a reasonable choice, as we expect the ``inert" ground state of the battery not to contribute to powering the implementation of NEPGs. 
This choice allows us to neglect the $\Delta$-term \footnote{More precisely, $\Delta$ can be neglected for all \lq\lq smooth'' states as shown in App. C of the supplemental material.} in Eq.\eqref{eq:finalchoi}.\\

{\it Optimizing the precision with a fixed number of battery levels:--} \label{Sec:MinChoi-FIX-N} 
We begin our analysis by looking for battery states that achieve the minimum Choi infidelity under the restriction that they are supported on the first $N$ levels of the battery, i.e. $\beta_{n}=0$ for $n\geq N $. An exact optimization of \eqref{eq:finalchoi} based on a discrete sine Fourier transform (see App. D 
in~\cite{SM}) reveals that the (asymptotically in $N\rightarrow \infty$) optimal states in this case are of the form
\begin{equation} \label{Eq:OptNstate}
\ket{\beta^{(N)}}_B:=C_{N}\underset{n=1}{\overset{N-1}{\sum}}\sin(\frac{\pi n}{N} )\ket{n}_B, 
\end{equation}
where $C_{N}$ is a normalization factor. 
The associated value of the Choi infidelity is
\begin{equation} \label{Eq:optNUD}
\frac{\epsilon_C(\mathbf{\Phi}_{|\beta^{(N)}\rangle},\mathcal{V})}{|V_{01}|^{2}}  = 4
\sin^2(\frac{\pi}{2N})
(1 - 
\sin^2(\frac{\pi}{2N})
|V_{01}|^2).
\end{equation} To reach arbitrarily small values of the infidelity we need $\sin^2(\frac{\pi}{2N}) \rightarrow 0$, i.e. $N \rightarrow \infty$. In this limit, the optimality of $\ket{\beta^{(N)}}_B$ and Eq.~\eqref{Eq:optNUD} implies a bound on the number  $N$ of occupied levels that is necessary to achieve a desired precision $\epsilon_{C}$: 
\begin{equation}
    N \geq \frac{ \pi |V_{01}|}{\sqrt{\epsilon_C}}+o(1). \label{eq:bound0}
\end{equation}
This is the first main result of the paper: if we want to implement any NEPG with infidelity $\epsilon_C$ we need a battery with a number of levels scaling as $N \sim 1/ \sqrt{\epsilon_{C}}$.

{\it A new figure of merit in NEPG optimization:--}
Remarkably, the amplitudes characterizing the optimal solutions $\ket{\beta^{(N)}}_B$ are {\it smooth functions} in the large $N$ limit, i.e. they are such that the discrete second derivative in Eq.~\eqref{eq:finalchoi} is negligible [of order $O(N^{-4})$ = $O(\epsilon_C^2)$].
In the rest of the manuscript, when considering more complicated optimizations (namely, with constraints on the initial state of the battery) we simplify our analysis by introducing a scale parameter $\delta>0$ 
which naturally generalize the role of $N^{-1}$ in Eq.\eqref{Eq:OptNstate} for functions which have non compact support. More specifically, 
we take the following ansatz for the battery amplitudes
\begin{equation}
  \beta_n^{(\delta)} := C_{\delta}\psi(n\delta) ,  \quad C_{\delta} := \left(\underset{n}{\sum} \abs{\psi(n \delta)}^{2}\right)^{-\frac{1}{2}},  \label{eq:ansatz}
\end{equation}
with $\lVert \psi \rVert_{2}=1$,
 \( \abs{\psi'(x)}, \abs{\psi''(x)} < \infty \) \footnote{Divergence of $\psi''(x)$ in a finite number of points is also allowed. In most cases is best to find the optimum in $\psi \in L_{2}[0,\infty]$ and verify a posteriori that the sum converges to the integral.}. Here, $\psi$ defines the overall \lq\lq shape" of the battery amplitudes over the energy spectrum, and the parameter $\delta$ shrinks this shape with respect to the discrete ladder of energy levels spaced by $\omega$.
As a preliminary observation, if \( \psi(0) \neq 0 \) we have \( \epsilon_{C} \propto \delta \), while any function vanishing in zero guarantees \( \epsilon_{C} \propto \delta^{2} \). This confirms that under the ansatz \eqref{eq:ansatz} choosing $\beta_0 =0$, i.e. \( \psi(0) = 0 \), is a well-motivated choice (see App. C in~\cite{SM}
).
Under these assumptions the Choi infidelity takes an elegant analytical form in the $\delta \rightarrow 0$ limit
\begin{align} 
  \frac{\epsilon_C(\mathbf{\Phi}_{|\beta^{(\delta)}\rangle},\mathcal{V})}{|V_{01}|^2}  &=
   \sum_{n=0}^{\infty}|\beta_{n+1}^{(\delta)}  - \beta_n^{(\delta)}|^2 +O\left(\delta^{4} \right) \nonumber\\
  &=:  \delta^{2}\,  {\rm UD}(\psi) +O(\delta^{3})\label{eq:UDdef}
\end{align}
where we introduced the {\it Unitary Defect} (UD) 
\begin{align}\label{eq:UDdef2}
{\rm UD} (\psi) = \int_0^{\infty} |\psi'(x)|^2 \dd x.
\end{align} 
Remarkably, the UD is independent of the gate, making it a genuine and universal measure of the quality of a quantum battery for the implementation of high precision  ($\delta, \epsilon_C \rightarrow 0$) NEPGs.Note that for a mixed battery state $\beta_{B}= \sum_i p_i \ketbra{\beta_i}_B$, Eqs.(\ref{eq:UDdef},\ref{eq:UDdef2}) remain valid upon using the average unitary defect $\sum_{i} p_i \, UD(\psi_i)$. 
Thus, the averaged UD is a {\it bona fide} quantifier of the quality of the battery state for generic initial preparations.

{\it Optimizing the precision with limited resources:--} \label{Sec:Optimizing resources}
In the following, we analyse in detail how to optimize battery states when only a limited amount of a given resource is available \cite{Chiribella, QFIBound,UniversalTheoremTakaji2022}.
In the context of NEPGs, a resource is a function $\mathcal{R}$ of the battery state  that satisfies certain properties \cite{Chiribella, UniversalTheoremTakaji2022}, such as monotonicity under energy-preserving operations and regularity with respect to trace distance or fidelity. 
Typical resources are functions of the modulus of the battery state coefficients
\begin{equation} \label{Eq:Def-resource} 
    \mathcal{R}(|\beta\rangle_{B})=\underset{n=0}{\overset{\infty}{\sum}} R(|\beta_{n}|,n),
\end{equation}
with $|\partial_{|\beta_{n}|}R(|\beta_{n}|,n)| < \infty $. Considering resources of the form~\eqref{Eq:Def-resource} we are interested in the optimal state
\begin{equation}
 |\beta_{\mathcal{R}}\rangle_B \coloneqq  \underset{\mathcal{R}(|\beta\rangle)\leq R_0}{ \arg \min} \epsilon_C(\mathbf{\Phi}_{|\beta\rangle},\mathcal{V}).
\end{equation}
For sufficiently regular resources the optimal states can be expected to be regular as well, thus, we will use the ansatz in Eq.~\eqref{eq:ansatz} and replace the sum in Eq.~\eqref{Eq:Def-resource} with an integral
\begin{equation} \label{eq:contres}
 \sum_{n=0}^{\infty} R(|\beta_{n}^{(\delta)}|,n) = \int_0^{\infty} R\left(C_{\delta}|\psi(\delta x)|,x\right) \dd x +O(\delta).
\end{equation}
The search for the optimal state can then be cast as a variational problem. We want to find the function $\psi_{\mathcal{R}}$ which solves
\begin{align} \label{Eq:MinProb-UD}
   \min_{\psi} \;  {\rm UD} (\psi)
   \, \,\, \text{s.t.}\, \, \,  \mathcal{R}(\psi)\leq R_0,
\end{align}
where $R_0$ is the maximum amount of the desired resource available in the battery.
This can be approached with the standard methods from variational calculus --introducing a Lagrange multiplier $\lambda$ and minimizing the corresponding action

\begin{align}
\mathcal{S}_{\mathcal{R}}& = \int_0^{\infty} \dd x\, \big\{|\psi'(x)|^{2}+ \lambda R(\sqrt{\delta}|\psi(\delta x)|,x) \big\},
\label{UD-costriant}
\end{align}
with the constraints $ \psi(0) = 0$ and $ \lVert \psi \rVert_{2} = 1$, and where we used $C_{\delta}=\sqrt{\delta}+o(1) $.
 
Solving the variational problem gives us the optimal pure state and an achievable lower bound to the minimum Choi infidelity. 
In many relevant cases $\mathcal{R}$ is linear in the state, ruling out the possibility that mixed states can outperform pure ones. This possibility can also be excluded for certain convex resources, such as the Quantum Fisher Information; see App. E in~\cite{SM}
for details.

\begin{table*}
\begin{tabular}{ |p{3cm}||p{6cm}|p{4cm}| }
 \hline
 \quad \quad \,Target & \, \, \,Optimal Battery Preparation & \, \quad Asymptotic minimum \\
 \hline
\centering Fixed $\langle E\rangle$; Min UD  &\centering  \(\ket{\beta_{\langle E\rangle}}:=C_{E}\underset{n>0}{\overset{\infty}{\sum}} Ai(\frac{n\omega \eta'}{\langle E\rangle}-x_{0}) \ket{n}_{B} \) & \quad \,  \({\rm UD}_{min}(\langle E\rangle )= \eta^{2} \frac{\omega^{2}}{\langle E\rangle^{2}}\) \\
\hline
\centering Fixed $\langle E^2\rangle $; Min UD  & \centering \(\ket{\beta_{\langle E^2\rangle }}:=C_{E^{2}} \underset{n>0}{\overset{\infty}{\sum}} \psi_{1}\left(n \sqrt{\frac{3\omega^{2}}{2\langle E^2\rangle }}\right) \ket{n}_{B}\) & \quad \,  \, 
\({\rm UD}_{min}(\langle E^2\rangle )= \frac{9}{4}\frac{\omega^{2}}{\langle E^2\rangle } \) \\
\hline
\centering Fixed $N$; Min UD & \centering $ \ket{\beta^{(N)}}:=C_{N}\underset{n>0}{\overset{N-1}{\sum}} \sin(\frac{n \pi }{N})\ket{n}_B $ & \quad \, 
\({\rm UD}_{min}(N)=\frac{\pi^{2}}{N^{2}} \)\\
 \hline
\end{tabular}
\caption{Solutions for the minimization of UD for fixed battery resources. In the third column, we used the shorthand notation ${\rm UD}(\mathcal{R})$ to indicate the value of the unitary defect on the optimal preparations of the battery, i.e. ${\rm UD}(\psi_{\mathcal{R}})$ for $\mathcal{R} = \langle E\rangle, \langle E^2 \rangle, N$.
The values $\eta^{2} \approx1.888, \eta'\approx 1.557 $ are obtained from the integration of Airy related functions and $C_{\langle E\rangle},
C_{\langle E^2 \rangle},C_{N}$ are normalization coefficients (see App. E in~\cite{SM}
). }
\label{tab:1}
\end{table*}

{\it Mapping to 1D ground state problems:--}
It is worth noticing that for many resources the problem above is mathematically equivalent to finding the ground state of a 1D particle. 
For resources like the average energy $ \langle E\rangle := \mbox{Tr}[H_B \beta]$ and the average squared energy $ \langle E^2\rangle :=\mbox{Tr}[H_B^2 \beta] $ we have
$R_{\langle E\rangle}(|\psi(x)|,x) = \omega x|\psi(x)|^2$ and $R_{\langle E^2\rangle}(|\psi(x)|,x)= \omega^{2} x^{2}|\psi(x)|^2$, respectively, so using the vanishing boundary condition and integration by parts we get
\begin{align} \label{eq:ground1D}
\mathcal{S}_{\mathcal{R}}\hspace{-1mm}=\hspace{-1mm} \int_0^{\infty} \hspace{-2mm}\dd x\, \psi(x) \left[  -  \delta^{2} \frac{\dd^2}{\dd x^2}  +\lambda \left(\frac{x \omega}{\delta} \right)^{m_{\mathcal{R}}} \right] \psi^*(x),
\end{align}
where $m_{\mathcal{R}} = 1,2$ for $\mathcal{R} = \langle E\rangle, \langle E^2\rangle$, respectively. 
We observe that the variational problems corresponding to these two resources are equivalent to ground-state problems for a particle in a one-sided linear and quadratic potential, respectively.
The solutions are textbook \cite{sakurai2020modern} and
the optimal functions $\psi_{\langle E \rangle},  \psi_{\langle E^2 \rangle}$ are respectively given by $Ai(x)$, the Airy function of the first kind and $\psi_1(x)$, the first excited wave function of a Harmonic oscillator.
After fixing $\lambda(\delta)$ we obtain the optimal battery states, one for each problem, all summarized in Tab. \ref{tab:1} (see App. E in~\cite{SM}
).

From the table, we can immediately derive a new set of bounds that improve upon the results found in previous literature \cite{Chiribella, CataliticCoherence}. Indeed, since the minimum averaged (squared) energy is given, in the limit of high precision, by the first two cells of the third column of  Tab. \ref{tab:1}, we find the following saturable necessary conditions
\begin{align} \label{eq:bound1}
    & \langle E\rangle\geq \eta \frac{\omega |V_{01}|^{2}}{\sqrt{\epsilon_{C}}} + o\left(\frac{1}{\sqrt{\epsilon_{C}}} \right), \\
    &  \langle E^2\rangle\geq \frac{9}{4} \frac{\omega^{2} |V_{01}|^{2}}{\epsilon_{C}} + o\left(\frac{1}{\epsilon_{C}} \right) ,\label{eq:bound2}
\end{align}
where $\epsilon_{C}$ is a shorthand notation for $\epsilon_{C}(\mathbf{\Phi}_{\beta},\mathcal{V})$
and $\eta \approx 1.374$ can be calculated from integrals of Airy-related functions. 
In App. E in \cite{SM} we solved the problem of optimal initial state also for a battery prepared with fixed initial QFI. This is remarkable, as in contrast to energy and squared energy the QFI is strictly convex in the state. 
We do not report the results in the main text, since the resulting bound has already been obtained in  \cite{QFIBound} in terms of the worst case infidelity
 $\epsilon_{wc}$. However, the derivation is interesting and highlights that the optimal state of the battery for this task can be obtained by mapping the problem in the ground state problem of a harmonic oscillator.

{\it Application in real devices:--} The bounds in Eqs. \eqref{eq:bound0}, \eqref{eq:bound1}, \eqref{eq:bound2} are derived in the asymptotic limit $\epsilon_C \rightarrow 0$, i.e. for high precision gates where the ansatz \eqref{eq:ansatz} and the unitary \eqref{eq:usbUP} are proved to yield optimal results.
Therefore, there is no choice of interaction or state preparation that can outperform \eqref{eq:bound0}, \eqref{eq:bound1}, \eqref{eq:bound2}. In real devices, it is more insightful to reformulate the bounds in terms of the {\it intrinsic error} that accompanies  limited-resources computation, as
\begin{equation}
    \frac{\epsilon_{C}}{|V_{01}|^2 } \geq \max \left(\frac{ \eta^2 \omega^4 |V_{01}|^2}{\langle E\rangle^2}, \frac{9 \omega^2}{4 \langle E^2\rangle}, \frac{\pi^2}{N^2}\right)+ o\left(\epsilon_{C}\right),
\end{equation}
Where $\omega$ is the qubit gap energy and $\hbar=1$.
This tells us that, together with other typical sources of error that occur in quantum architectures (due to external noise, spurious excitations, undesired fluctuations in the control lines) there is a contribution to infidelity coming from the limited resources in the control pulses.

In many practical applications the driving is performed via semi-classical coherent pulses, that we can represent with a 
coherent battery state $\ket{\beta}_B= \sum_{n=1}^{\infty} \frac{\alpha^n}{n!}\ket{n}e^{-|\alpha|^2/2}$. 
By computing the precision achieved in this case by plugging the coherent state in Eq. \eqref{eq:finalchoi} one finds that the scaling of gate-implementation error with energy is sub-optimal 
$\epsilon_C \geq  \frac{ \omega |V_{01}^2|}{4 \langle E\rangle }$ with $\langle E \rangle = \omega |\alpha|^2$.
This suggests that a quantum battery approach could enhance the energy/precision trade-off if compared to semi-classical pulses
(See App. E.4 in \cite{SM} for details).


 A comment is mandatory on the comparison of the interaction unitary in Eq. \eqref{eq:usbUP} with more experimentally feasible models.
 In particular, note that  $U_{SB}$ has a ``flat interaction strength'', i.e. $U^{(n)}_{SB}$ is the same in all subspaces labelled by $n$. This highlights an important limitation of natural interaction models. For instance, the Jaynes-Cummings Hamiltonian $H_{JC} \propto \sum_n  \sqrt{n}(\ketbra{n+1}{n}_B\otimes \ketbra{0}{1}_S + h.c.)$ suffers from the $\sqrt{n}$ prefactor.  The latter has a negligible effect $\sqrt n \approx |\alpha|\big (1+O(1/|\alpha|)\big)$ for bright coherent cavity-modes $\ket{\alpha}$, but becomes the bottleneck when minimizing the energy of the battery.
 In this context, the engineering of intensity-dependent Jaynes–Cummings interactions \cite{IntensityDependentJC2022} emerges as a promising energy-conserving interaction model, and could attract renewed interest.

{\it Conclusions :--}
We studied the implementation of non energy-preserving unitary gates powered by external batteries. We proved that, for high-precision gate implementation, a new quantity, dubbed the {\it Unitary Defect}, is fundamentally linked to precision.
Using the UD, we optimized the initial state of the battery in various scenarios, including maximizing the precision for given initial battery resources and for batteries with a fixed number of levels.
Our approach yields three new bounds that relate implementation precision to the battery’s initial average energy, its average squared energy, and the number of accessible energy levels. These results improve upon previous bounds \cite{Chiribella, CataliticCoherence} and establish a general framework for efficiently optimizing the battery state.
Our results find natural implications in quantum thermodynamics and computing. In the active topic of quantum batteries~\cite{Campaioli2024,BatteryCap}, our work places the unitary defect as a complementary figure of merit to the ergotropy --a measure of the extractable work from a quantum state~\cite{allahverdyan2004maximal}. Likewise, the unitary defect becomes also a relevant quantity to assess the thermodynamic cost of  quantum gates~\cite{Chiribella,Auffeves2022}.

These findings pave the way to further research. A promising direction is the generalization to multi-qubit systems, where finding the optimal battery Hamiltonian and interactions presents substantial challenges due to the presence of spectral degeneracies in the Hamiltonian. However, in practical protocols where multi-qubit operations are often performed by a sequence of pulses addressing a single pair of levels at the time, our formalism naturally applies. In particular, in App. F.1 in \cite{SM} we show that in this case the Choi infidelity is given by \eqref{eq:UDdef} with the rhs multiplied by $\frac{2}{d}$, where $d$ is the total number of levels. Another interesting direction of research is to further investigate the interplay between the intrinsic error due to limited resources in the implementation process and other sources of error, and to precisely assess its relevance in practical devices. This would be particularly important for technological applications, especially when the precision–energy trade-off becomes critical \cite{silva2023classical}. 


{\it Acknowledgments :--}
RC and PS acknowledge support from  Swiss National Science Foundation (NCCR SwissMAP).
VC and VG acknowledge financial support by MUR (Ministero dell’ Università e della Ricerca) through the PNRR MUR project PE0000023-NQSTI. M.P.-L. acknowledges support from the Grant RYC2022-036958-I funded by MICIU/AEI/10.13039/501100011033 and by ESF+.

\bibliography{biblio}

\appendix
\begin{widetext}
\section{Choi infidelity} \label{sec:Choi-Infidelity}
The Choi infidelity we adopt in the manuscript is formally defined as 
\begin{equation} \label{eq:CHOIinfdef}
  \epsilon_{C}(\mathbf{\Phi}_{\beta}, \mathcal{V}):=
  1-F(\sigma_{\mathbf{\Phi}_{\beta}},\sigma_{\mathcal{V}}),
\end{equation}
where given  $\sigma$ and $\tau$  density matrices 
$F(\sigma,\tau):=\left(\tr\left[\sqrt{\sqrt{\sigma}\tau\sqrt{\sigma}}\right]\right)^{2}$ is the associated fidelity function, and where 
 $\sigma_{\mathbf{\Phi}_{\beta}}$ and $\sigma_{\mathcal{V}}$ are the 
 Choi states of the CTPT map $\mathbf{\Phi}_{\beta}$ and of the unitary gate $\mathcal{V}$ respectively, i.e. the density matrices \begin{eqnarray} 
\sigma_{\mathbf{\Phi}_{\beta}} &:=& \left( \mathcal{I}_{X} \otimes \mathbf{\Phi}_{\beta} \right)(\ketbra{\Psi_{\max}}{\Psi_{\max}})=\sum_{k} {K_S^{(k)}} \ketbra{\Psi_{\max}} K_S^{(k)\dag} \;, \\ 
\sigma_{\mathcal{V}} &:=& \left( \mathcal{I}_{X} \otimes \mathcal{V} \right)  (\ketbra{\Psi_{\max}}{\Psi_{\max}})
={V}_S\ketbra{\Psi_{\max}}{\Psi_{\max}}V^\dag_{S}, \label{defCHOIS} 
\end{eqnarray} 
 with $\ket{\Psi_{\max}}:=\frac{1}{\sqrt{d}}\underset{j=0}{\overset{d-1}{\sum}}\ket{j}_{X}\otimes \ket{j}_{S}$  a maximally 
entangled state of
the joint system formed by $S$ and by an auxiliary, $d$-dimensional space $X$,  and where 
$\mathcal{I}_{X}$ is the identity channel on $X$. 
We remark that the specific  choice of the local basis $\{ \ket{j}_{S}\}$ and $\{ \ket{j}_{X}\}$  which enters in the
definitions of  (\ref{defCHOIS}) 
  is  irrelevant in the evaluation of Eq.~(\ref{eq:CHOIinfdef}) and that
  Eq.~(2) directly follows observing that $\langle \Psi_{\max}| V_S^\dag   {K_S^{(k)}} |\Psi_{\max}\rangle =
  \mbox{Tr}[  V_S^\dag   {K_S^{(k)}}] /d$. 
 We also point out that $\epsilon_{C}(\mathbf{\Phi}_{\beta}, \mathcal{V})$  can be expressed in terms of expectation values with respect to the battery input state $\beta_B$ of a positive semidefinite operator $M_B$ that depends upon the interaction operator
 $U_{SB}$ and the target gate $V_S$. Specifically   we can write 
\begin{eqnarray} 
F(\sigma_{\mathbf{\Phi}_{\beta}},\sigma_{\mathcal{V}}) &=&\langle \Psi_{\max}| V^\dag_S \; \big( 
\mathcal{I}_{X} \otimes \mathbf{\Phi}_{\beta}\big)(\ketbra{\Psi_{\max}}{\Psi_{\max}}) \;  {V}_{S}  |\Psi_{\max}\rangle \nonumber \\
&=&  
\langle \Psi_{\max}|V^\dag_S  \; \mbox{Tr}_B[  U_{SB}  (\ketbra{\Psi_{\max}}{\Psi_{\max}}\otimes \beta_B) U^\dag_{SB}] \; {V}_{S}  |\Psi_{\max}\rangle
=\mbox{Tr}[ \beta_B M_B], 
\end{eqnarray} 
with 
\begin{eqnarray} 
M_B &:=& \langle \Psi_{\max}| U^\dag_{SB} {V}_{S}  |\Psi_{\max}\rangle\langle \Psi_{\max} | {V}^\dag_{S}U_{SB}  |\Psi_{\max}\rangle=Q_B^\dag Q_B, \\
Q_B&:=&\langle \Psi_{\max} | {V}^\dag_{S}U_{SB}  |\Psi_{\max}\rangle = \frac{1}{d} \sum_{j=0}^{d-1} {_S\langle} j| 
{V}^\dag_{S}U_{SB} |j\rangle_S =  \frac{1}{d} \sum_{j,j'=0}^{d-1} {_S\langle} j'| 
U_{SB} |j\rangle_S \; V^*_{j'j} , \label{defQ}
\end{eqnarray} 
with $V_{jj'} :=  {_S\langle} j|V_S|j'\rangle_S$ the matrix elements of $V_S$ w.r.t. to 
orthonormal  basis $\{ \ket{j}_{S}\}$  which enter in the definition of $\ket{\Psi_{\max}}$. Accordingly, we can write
  \begin{equation} \label{eq:CHOIinfdefrispQ}
  \epsilon_{C}(\mathbf{\Phi}_{\beta}, \mathcal{V})=1- \mbox{Tr}[ \beta_B M_B] = 
 \mbox{Tr}[ \beta_B ({I}_B-M_B)].
\end{equation}
As observed in the main text, by linearity it follows that minimization of $\epsilon_{C}(\mathbf{\Phi}_{\beta}, \mathcal{V})$ for fixed choices of $V_S$ and $U_{SB}$, can be restricted to only pure states of the battery. 
In particular, such minimum is attained by the eigenstate $|\bar{\beta}\rangle_B$ of $M_B$ with maximum eigenvalue $\lambda_{\max}$, i.e. 
  \begin{eqnarray} \label{eq:CHOIinfdefrispQmin}
 \min_{\beta_B}  \epsilon_{C}(\mathbf{\Phi}_{\beta}, \mathcal{V}) = 
  \min_{|\beta\rangle_B\in {\cal H}_B}  \epsilon_{C}(\mathbf{\Phi}_{|\beta\rangle}, \mathcal{V})  &=&  \epsilon_{C}(\mathbf{\Phi}_{|\bar{\beta}\rangle}, \mathcal{V})  =1- \lambda_{\max} , \qquad
  M_B |\bar{\beta}\rangle_B =\lambda_{\max} |\bar{\beta}\rangle_B. 
\end{eqnarray}
We conclude recalling that for the special case where $S$ is a qubit and $B$ is a harmonic oscillator 
the most general unitary interaction $U_{SB}$ that is energy preserving, i.e. $[U_{SB},H_{S}+H_{B}]=0$ can be written as 
\begin{equation} 
{U}_{SB}  = \underset{S}{\ketbra{0}{0}} \otimes \underset{B}{\ketbra{0}{0}}
 +\underset{n>0}{\overset{\infty}{\bigoplus}}\; {U}^{(n)}_{SB}, \qquad   \qquad {U}^{(n)}_{SB}  := \sum_{i,j=0,1} 
 U^{(n)}_{ij} \underset{S}{\ketbra{i}{j}} \otimes \underset{B}{\ketbra{n-i}{n-j}}, \label{eq:usbUPgen}
\end{equation}
where for each $n>0$ the coefficients $U^{(n)}_{ij}$ define  $2\times2$ unitary matrices, and $\{ |i\rangle_S\}$ and $\{ |n\rangle_B\}$ are the energy eigenbasis of $S$ and $B$ defined in 
Eq.~(4). Under these conditions Eq.~(\ref{defQ}) becomes
\begin{eqnarray} 
Q_B&=&\frac{1}{2} \Big( V_{00} \;  \underset{B}{\ketbra{0}{0}} 
+   \sum_{n > 0} \sum_{j,i=0,1} U^{(n)}_{ij}  \; V^*_{ij}  \; \underset{B}{\ketbra{n-i}{n-j}}  \Big). \label{defQ1}
\end{eqnarray}

\subsection{Equivalence of worst-case infidelity and Choi infidelity} \label{sec:appbound}

In much of the literature on the implementation of asymmetric unitary gates, the results are stated in terms of the worst-case infidelity defined as 
\begin{eqnarray} \label{eq:WCinfdef}
   \epsilon_{wc}(\mathbf{\Phi}_{\beta}, \mathcal{V})&:=&1-\underset{\ket{\Psi}\in \mathcal{H}_{XS}}{\min}F((\mathcal{I}_{X}\otimes \mathbf{\Phi}_{\beta})(\ketbra{\Psi}), (\mathcal{I}_X \otimes \mathcal{V})  (\ketbra{\Psi}{\Psi})) \\
   \nonumber &=&1 - \underset{\ket{\Psi}\in \mathcal{H}_{XS}}{\min} \braket{\Psi|V_S^\dag \mathbf{\Phi}_\beta ({\ketbra{\Psi}{\Psi}})V_S |\Psi},
\end{eqnarray}
 For finite dimensional systems this quantity 
 is equivalent to the Choi infidelity $\epsilon_{C}(\mathbf{\Phi}_{\beta}, \mathcal{V})$. 
 More specifically we observe that 
\begin{prop} Let $\mathbf{\Phi}$ be an arbitrary CPTP channel defined on
an Hilbert space of dimension $d$ and $\mathcal{V} := V\cdots V^{\dag}$  a unitary mapping on the same space,
 then 
\label{Prop:ChoiInfidelityineq}
    \begin{equation} \label{wcVSChoiFidapp}
\frac{\epsilon_{wc}(\mathbf{\Phi},\mathcal{V})}{d}\leq \epsilon_{C}(\mathbf{\Phi},\mathcal{V})\leq \epsilon_{wc}(\mathbf{\Phi},\mathcal{V}).
\end{equation}
\end{prop}

\begin{proof}
   The inequality $ \epsilon_{C}(\mathbf{\Phi},\mathcal{V})\leq \epsilon_{wc}(\mathbf{\Phi},\mathcal{V})$ is trivial since the Choi infidelity is obtained by choosing a particular value of $\ket{\Psi}$ - namely, the totally entangled state $\ket{\Psi_{\max}}$  that appears in the definition of the Choi state - in Eq. \eqref{eq:WCinfdef}. To prove the other inequality, consider
 $\ket{\chi}$ the worst-case state in Eq. \eqref{eq:WCinfdef}, i.e., 
 \begin{eqnarray} 
 1 -  \braket{\chi|V_S^\dag \mathbf{\Phi}({\ketbra{\chi}})V_S |\chi}= \epsilon_{wc}(\mathbf{\Phi},\mathcal{V}).
 \end{eqnarray}  By Schmidt decomposition it can be written as 
\begin{equation}
    \ket{\chi} := \sum_{\ell=0}^{d-1} \sqrt{P_{\ell}}\ket{\ell}_{X}\otimes \ket{\ell}_{S}, \end{equation} 
 with $\{ P_{\ell}\}$ probabilities, and with $\{ \ket{\ell}_{X}\}_{\ell}$ and  $\{ \ket{\ell}_{S}\}_{\ell}$
  properly choosing the local bases for $X$ and $S$. 
To continue to proof, we have to add yet another auxiliary space $R$ and introduce channel 
\begin{eqnarray} 
\mathcal{S}(\cdots):= F_{S}^{(0)}(\cdots) F_{S}^{(0) \dagger}\otimes \underset{R}{\ketbra{0}{0}} +F_{S}^{(1)}(\cdots) F_{S}^{(1)\dagger}\otimes \underset{R}{\ketbra{1}{1}}\;, \end{eqnarray}  with $|0\rangle_R$,  $|1\rangle_R$ orthonormal vectors of $R$ and with
\begin{align}
 F_{S}^{(0)}:=\sqrt{q_0 \; {\rho}_S}\;,   \quad \quad
   F_{S}^{(1)}:= \sqrt{{I}_S -F_{S}^{(0)\dag}F_{S}^{(0)}} =\sqrt{{I}_S - q_0 \rho_S} \;,
\end{align}
where 
$q_{0}:=(\underset{j}{\max} P_{j})^{-1}\geq 1$ and $\rho_S: =\sumab{\ell=0}{d-1}{P_{\ell}} \underset{S}{\ketbra{\ell}{\ell}}$ is the reduced density matrix with respect to $S$ of the state $\ket{\chi}$. 
Notice that the channel acts trivially on $X$ and that the chosen value of $q_{0}$ is the largest such that ${I}_S -F_{S}^{(0)\dag}F_{S}^{(0)}\geq 0$, and that 
$1\geq {q_0}/{d}$. 
Identifying the maximally entangled  $\ket{\Psi_{\max}}$ which enters in Eq.~(\ref{defCHOIS}) with 
$\frac{1}{\sqrt{d}}\underset{\ell=0}{\overset{d-1}{\sum}}\ket{\ell}_{X}\otimes \ket{\ell}_{S}$, we observe that
\begin{equation}
 \mathcal{S}({\ketbra{\Psi_{\max}}{\Psi_{\max}}})=\frac{q_{0}}{d}{\ketbra{\chi}{\chi}}\otimes\underset{R}{\ketbra{0}{0}}+(1 -\frac{q_{0}}{d}) {\ketbra{\chi'}{\chi'}}\otimes \underset{R}{\ketbra{1}{1}},
\end{equation}
for some pure state $ \ket{\chi'}$ of $XS$.
We use now the monotonicity of the fidelity under quantum channels and the commutativity of $\mathcal{S}$ with $\mathbf{\Phi}$ and ${\cal V}$ to prove the following chain of relations
\begin{eqnarray}
  F(\sigma_{\mathbf{\Phi}},\sigma_{\mathcal{V}})&\leq& F\left(\mathcal{S}(\sigma_{\mathbf{\Phi}}),\mathcal{S}(\sigma_{\mathcal{V}})\right ) =  F\left(\mathbf{\Phi}\circ \mathcal{S}(\ketbra{\Psi_{\max}})),\mathcal{V}\circ {\cal S}(\ketbra{\Psi_{\max}})\right )\\
  &=&\nonumber F\left(\frac{q_{0}}{d}\mathbf{\Phi}(\ketbra{\chi}) \otimes \underset{R}{\ketbra{0}{0}}+(1 -\frac{q_{0}}{d}) \mathbf{\Phi}({\ketbra{\chi'}{\chi'}})\otimes \underset{R}{\ketbra{1}{1}},\frac{q_{0}}{d}{{\cal V}(\ketbra{\chi})}\otimes \underset{R}{\ketbra{0}{0}}+ (1 -\frac{q_{0}}{d}) {{\cal V}(\ketbra{\chi'})}\otimes\underset{R}{\ketbra{1}{1}})\right )\\
  &=&\nonumber\left( \frac{q_{0}}{d}\sqrt{\braket{\chi|V_S^\dag \mathbf{\Phi}({\ketbra{\chi}})V_S |\chi}}+(1 -\frac{q_{0}}{d})\sqrt{ \braket{\chi'|V_S^\dag \mathbf{\Phi}({\ketbra{\chi'}})V_S |\chi'}}\right)^2
    \\
    &\leq& \nonumber \frac{q_{0}}{d}{\braket{\chi|V_S^\dag \mathbf{\Phi}({\ketbra{\chi}})V_S |\chi}}+(1 -\frac{q_{0}}{d}){ \braket{\chi'|V_S^\dag \mathbf{\Phi}({\ketbra{\chi'}})V_S |\chi'}} \\
&\leq& \nonumber \frac{q_{0}}{d}{\braket{\chi|V_S^\dag \mathbf{\Phi}({\ketbra{\chi}})V_S |\chi}}+(1 -\frac{q_{0}}{d}) 
= 1 -\frac{q_{0}}{d} \left(1- {\braket{\chi|V_S^\dag \mathbf{\Phi}({\ketbra{\chi}})V_S |\chi}}\right) =
 1 -\frac{q_{0}}{d} \epsilon_{wc}(\mathbf{\Phi},\mathcal{V})\;,   \end{eqnarray}
where $\sigma_{\mathbf{\Phi}}$ and $\sigma_{\mathcal{V}}$ are the Choi states of the map $\mathbf{\Phi}$
and ${\cal V}$   and where we invoked the Cauchy-Swartz inequality. 
Hence, we conclude
\begin{align*}
  \epsilon_{C}(\mathbf{\Phi},\mathcal{V})=&1-F(\sigma_{\mathbf{\Phi}},\sigma_{\mathcal{V}}) \geq \frac{q_{0}}{d} \epsilon_{wc}(\mathbf{\Phi},\mathcal{V}) \geq \frac{\epsilon_{wc}(\mathbf{\Phi},\mathcal{V})}{d},
\end{align*}
that is the desired result.
\end{proof}

The bound $\epsilon_{C}(\mathbf{\Phi},\mathcal{V})\leq \epsilon_{wc}(\mathbf{\Phi},\mathcal{V})$ is easily seen to be saturable. For instance consider $\mathcal{V}=\mathcal{I}$ and $\mathbf{\Phi}(\cdots) =\rho_{\star} \mbox{Tr} [ \cdots]$ with  some fixed output state $\rho_{\star}$  (a trash-and-replace channel).
 To illustrate the tightness of the bound $\frac{\epsilon_{wc}(\mathbf{\Phi},\mathcal{V})}{d}\leq \epsilon_{C}(\mathbf{\Phi},\mathcal{V})$ we consider the following example with $\mathbf{\Phi}=\mathcal{I}$ and $\mathcal{V}$ associated with the unitary transformation 
\begin{equation}
 {V}_S= e^{i \varphi}\underset{S}{\ketbra{0}} + e^{-i \varphi}\underset{S}{\ketbra{1}}+ \sum_{i=2}^{d-1} \underset{S}{\ketbra{i}}.
\end{equation}
In this case one finds 
\begin{equation}
    F(\sigma_{\mathbf{\Phi}},\sigma_{\mathbf{\mathcal{V}}}) = \left|\frac{1}{d} (\bra{00}e^{i\varphi}\ket{00}+\bra{11}e^{-i\varphi}\ket{11}+\sum_{i\geq 2}^{d-1} \braket{ii}) \right|^2 = \frac{(d-2 (1-\cos (\varphi )))^2}{d^2},
\end{equation}
and $\epsilon_{C}(\mathbf{\Phi},\mathcal{V}) =1-  F(\sigma_{\mathbf{\Phi}},\sigma_{\mathbf{\mathcal{V}}})= 4 \frac{ 1-\cos (\varphi) }{d} \left(1-\frac{1-\cos (\varphi )}{d}\right)$. In turn computing $\epsilon_{wc}(\mathbf{\Phi},\mathcal{V})$ correspond to the optimal discrimination task between the unitary channels $\mathcal I$ and $\mathcal{V}$. The minimal achievable fidelity $F_{wc}$ for this task is well know~\cite{PhysRevA.66.052107}. It does not require entanglement with the ancillary and is attained by probing the channels with the state $\ket{+}_S = \frac{1}{\sqrt 2}(\ket{0}+\ket{1})$, and gives
\begin{equation}
F_{wc} = \left|\frac{\bra{0}e^{i \varphi} \ket{0} + \bra{1}e^{-i \varphi}\ket{1}}{2}\right|^2 = \cos^2(\varphi)\quad \implies \quad \epsilon_{wc}(\mathbf{\Phi},\mathcal{V})=   \sin^2\left(\varphi\right).
\end{equation}
For the ratio of infidelities we found an example of channels with 
\begin{equation}
\frac{\epsilon_{wc}(\mathbf{\Phi},\mathcal{V})}{\epsilon_{C}(\mathbf{\Phi},\mathcal{V})} = \frac{\sin^2(\varphi)}{4 \frac{ 1-\cos (\varphi) }{d} \left(1-\frac{1-\cos (\varphi )}{d}\right)} =\frac{d^2 (\cos (\phi )+1)}{4 (d+\cos (\phi )-1)} \to \frac{d}{2}
\end{equation}
in the limit $\varphi \to 0$. We have thus shown that the dependence of the bound in the dimension is correct up to at most a constant factor of $\frac{1}{2}$. In other words, for any dimension $d>2$ and any $\epsilon>0$ we have found two channels with $\frac{\epsilon_{wc}(\mathbf{\Phi},\mathcal{V})}{\epsilon_{C}(\mathbf{\Phi},\mathcal{V})} =  (\frac{d}{2}-\epsilon)$, while for $d=2$ we find $\frac{\epsilon_{wc}(\mathbf{\Phi},\mathcal{V})}{\epsilon_{C}(\mathbf{\Phi},\mathcal{V})}  =1$.

\section{Calculation of the Choi infidelity}
\label{sec:calcchoi}

We start by producing a Kraus set  for the CPTP map $\mathbf{\Phi}_{\beta}$ associated with  the global unitary evolution in Eq.~\eqref{eq:usbUP}. For the case where  the battery is initialized in the pure state $|\beta\rangle_B = \underset{n=0}{\overset{\infty}{\sum}} \beta_n |n\rangle_{B}$, with $|n\rangle_B$ being the energy eigenstates of $H_B$, we can write 
 \begin{eqnarray} K_S^{(n)} &=& {_B\langle} n | U_{SB} |\beta\rangle_B = \sum_{m=0}^{\infty} \beta_m\;  {_B\langle}  n | U_{SB} |m\rangle_B\;. 
 \label{defKSn} \end{eqnarray} 
For $n=0$ this gives 
\begin{equation} 
{K}_S^{(0)} = \beta_0 \left( \underset{S}{\ketbra{0}{0}}
   + V_{11}  
   \underset{S}{\ketbra{1}{1}}\right) + V_{10} \beta_1
   \underset{S}{\ketbra{1}{0}}. \label{eq:kraus0}
\end{equation}
For $n \geq  1$ instead we get 
\begin{align} 
K_S^{(n)} & = \sum_{m=0}^{\infty} \beta_m 
\sum_{i,j=0,1} \sum_{l=1}^{\infty} 
\underset{B}{\langle} n\underset{B}{\ketbra{l-i}{l-j}} m \underset{B}{\rangle} \; \underset{S}{\ketbra{i}{j}} V_{ij} 
= \sum_{m=0}^{\infty} \sum_{l=1}^{\infty} \sum_{i,j=0,1} \beta_m \delta_{n,l-i} \delta_{m,l-j}
\underset{S}{\ketbra{i}{j}} V_{ij}  \\ & = \sum_{m=1-j}^{\infty} \sum_{i,j=0,1} \beta_m \delta_{n+i,m+j}
\underset{S}{\ketbra{i}{j}} V_{ij}  \notag
   = \sum_{i,j=0,1} \beta_{n+i-j} 
\underset{S}{\ketbra{i}{j}} V_{ij}  =
  \beta_{n}V_{S}+\sum_{i,j=0,1} (\beta_{n+i-j}-\beta_{n}) 
\underset{S}{\ketbra{i}{j}} V_{ij},
\end{align}
where in the last identity we sum and subtract the  term $\beta_{n}V_{S}$.

To compute Choi infidelity \eqref{eq:CHOIinfdef}, we have to
multiply all Kraus operators found above by $V^{\dag}_S$. We obtain
\begin{align} \notag 
 V^{\dag}_SK_S^{(n)}   & = \beta_{n}  \;   {I}_S+ 
 \sum_{k,i,j} V^{\dagger}_{ik}V_{kj}(\beta_{n+k-j} -\beta_n)\underset{S}{\ketbra{i}{j}}=: \beta_{n} \;  {I}_S+ J_S^{(n)},
\end{align}
where the second to last equality is a definition of the operator $J_S^{(n)}$, which implies 
\begin{equation}
\abs{\mbox{Tr}[V^{\dag}_{S}K_S^{(n)}]}^{2}=|\beta_{n}|^{2}|\mbox{Tr}[ {I}_S]|^{2}+ 4\Re[\mbox{Tr}[\beta_{n}^{*}J_S^{(n)}]]+\abs{\mbox{Tr}[J_S^{(n)}]}^{2}.
\end{equation}
Direct computation returns $\mbox{Tr}[J_S^{(n)}]=\abs{V_{01}}^{2}(\beta_{n+1}+\beta_{n-1}-2\beta_{n})$, plugging in Eq.(2)  and using the normalization condition $\underset{n>0}{\overset{\infty}{\sum}}\abs{\beta_{n}}^{2}=1-\abs{\beta_{0}}^{2}$ we get
\begin{equation} 
\epsilon_C(\mathbf{\Phi}_{\beta},\mathcal{V}) =  |V_{01}|^2 \sum_{n>0} \big[ - \Re[\beta_{n}^{*}(\beta_{n+1}+\beta_{n-1}-2\beta_{n})]  - \frac{|V_{01}|^2}{4}   |\beta_{n+1} + \beta_{n-1} - 2 \beta_n|^2 \big] + \abs{\beta_{0}}^{2}- \frac{1}{4}|\mbox{Tr}[V_{S}^{\dagger}K_{S}^{(0)}]|^2.
\label{eq:finalchoi_app}
\end{equation}
 
To simplify further, we use summation by parts:

\begin{equation}
    \sum_{n=1}^{\infty} f_n [g_{n+1} - g_{n}]
    = - f_1 g_1 - \sum_{n=1}^{\infty} g_{n+1} [f_{n+1} - f_n],
\end{equation}

with $f_n = \beta^{*}_n$ and $g_n = \beta_n - \beta_{n-1}$.
Thus we obtain

\begin{equation}
    \sum_{n=1}^{\infty} \beta^{*}_n [\beta_{n+1} +  \beta_{n-1} - 2\beta_{n}]
    = - \beta^{*}_1 (\beta_1 - \beta_{0}) - \sum_{n=1}^{\infty} [\beta_{n+1} - \beta_{n}] [\beta_{n+1}^* - \beta_n^*].
\end{equation}

If we replace the equation above inside Eq. \eqref{eq:finalchoi_app} we obtain Eq.
(6) in the main text, after defining

\begin{align} \notag
    \Delta(\beta_0,\beta_1, V_S) & = |\beta_0|^2 + |V_{01}|^2 |\beta_1|^2 -   |V_{01}|^2 \Re[\beta_0 \beta_1^*]
    - \frac{1}{4} \big| \mbox{Tr}[V_S^{\dagger} K_S^{(0)}]\big|^2 
    \\
    & =
 |\beta_0|^2 + |V_{01}|^2(|\beta_1|^2 - \Re[\beta_0 \beta_1^*])
    - \frac{1}{4} \big| \beta_0 (V_{00}^{*}+|V_{11}|^2) + \beta_1 |V_{01}|^2   \big|^2.
\end{align}

\section{Comments on the assumptions on the battery input state and
the global interaction}

\subsection{Analysis of the ground state of the battery}
\label{sec:calcchoi2}

The aim of this section is to prove that to obtain an optimal state of the battery it is necessary to have a ground state population that is sufficiently small.
In terms of the function $\psi$ defined in Eq. 
(10) of the main text,
this corresponds to choosing $\psi(0)=0$. We first establish a useful lemma on the Choi infidelity:

\begin{lem} \label{lem:Choi-Inf-Kraus}
    Let $\mathbf{\Phi}_{\beta}$  be a CPTP map acting on a $d$-dimensional space described by the Kraus set $\{ K_{S}^{(k)}\}_k$, and ${\cal V}$ a unitary evolution associated with the gate $V_S$. The Choi infidelity 
     $\epsilon_{C}(\mathbf{\Phi}_\beta,\mathcal{V})$ can be expressed as 
\begin{equation}\label{decomposizione1} 
    \epsilon_{C}(\mathbf{\Phi}_\beta,\mathcal{V})=  \underset{k}{\sum} \epsilon^{(k)}_C(\mathbf{\Phi}_\beta,\mathcal{V}),
\end{equation}
with $\epsilon^{(k)}_C(\mathbf{\Phi}_\beta,\mathcal{V})$ non-negative quantities defined as
  {\rm \begin{equation} \label{defepsilonk} 
      \epsilon^{(k)}_C(\mathbf{\Phi}_\beta,\mathcal{V}):=\frac{1}{d}\left( \mbox{Tr}[K_{S}^{(k)\dagger}K_{S}^{(k)}]-\frac{\abs{\mbox{Tr}[V_S^\dag K_{S}^{(k)}]}^{2}}{d}\right).
    \end{equation}}
\end{lem}
\begin{proof}
    Because of the normalization condition $\underset{k}{\sum} K_{S}^{(k)\dagger}K_{S}^{(k)}= I_S$ we have: 
    \begin{align}
        \epsilon_{C}(\mathbf{\Phi}_\beta,\mathcal{V})=1-\frac{1}{d^{2}}\underset{k}{\sum} 
        \abs{\mbox{Tr}[V_S^\dag K_{S}^{(k)}]}^{2} 
        =\frac{1}{d}\underset{k}{\sum}\left(\mbox{Tr}[K_{S}^{(k)\dagger}K_{S}^{(k)}]-\frac{ \abs{\mbox{Tr}[V_S^\dag K_{S}^{(k)}]}^{2} }{d} \right),
    \end{align}
    which proves Eq.~(\ref{decomposizione1}). 
The positivity of   $\epsilon^{(k)}_C(\mathbf{\Phi}_\beta,\mathcal{V})$ follows instead from the Cauchy-Schwarz inequality which for any operator $K_{S}^{(k)}$ allows us to write
    \begin{align}
     &  \abs{\mbox{Tr}[V_S^\dag K_{S}^{(k)}]}^{2} \leq  \mbox{Tr}[ V_S  V_S^\dag] \mbox{Tr}[K_{S}^{(k)\dag}K_{S}^{(k)} ]=d \, \mbox{Tr}[K_{S}^{(k)\dag}K_{S}^{(k)} ].
    \end{align}
\end{proof}

Now we will prove a theorem stating that, for the qubit-battery model defined in Eq.~(4) of the main text, populating the ground state implies a high Choi infidelity.

\begin{thm} \label{Th:BoundaryOcc}
Let $\mathbf{\Phi}_{|\beta\rangle}$ be the CPTP map defined in Eq.~(1) 
of the main text, corresponding to the qubit-battery model in Eq.~(4)
, and associated with a battery input state $|\beta\rangle_B$. Let $\beta_0 := {_B\langle}0|\beta\rangle_B$ denote the population amplitude of the battery ground state. Then, for all
choices of the NEP unitary coupling $U_{SB}$, we can identify  a function $\eta(V_S)$ such that
\begin{align}\label{thesis0} 
    \epsilon_{C}(\mathbf{\Phi}_{|\beta\rangle},\mathcal{V}) \geq \epsilon^{(0)}_C(\mathbf{\Phi}_{|\beta\rangle},\mathcal{V}) \geq \eta(V_{S}) \abs{\beta_{0}}^{2},  \qquad 
 [V_{S},H_{S}]\neq 0 \implies \eta(V_{S})>0.
\end{align}

\end{thm}

\begin{proof}
The result is a consequence  of Lemma~\ref{lem:Choi-Inf-Kraus} which thanks to the positivity of the 
$ \epsilon^{(k)}_C(\mathbf{\Phi}_{|\beta\rangle},\mathcal{V})$'s allows us to establish the 
 following lower bound for the Choi infidelity of the model
\begin{equation}\label{decomposizione1ineq} 
    \epsilon_{C}(\mathbf{\Phi}_{|\beta\rangle},\mathcal{V}) \geq  \epsilon^{(0)}_C(\mathbf{\Phi}_{|\beta\rangle},\mathcal{V}),
\end{equation}
obtained by eliminating all the terms on the right-hand-side of Eq.~(\ref{decomposizione1ineq}) but the first. 
Recall that for the qubit-battery model the most general NEP unitary coupling has the form of Eq.~\eqref{eq:usbUPgen}. 
Writing the Kraus operator $K_{S}^{(0)}$ as in Eqs.~\eqref{eq:kraus0} we have 
\begin{eqnarray}
K_{S}^{(0)}  &=& {_B\langle} 0| U_{SB} |\beta\rangle_B =  \beta_0 \left( 
 \underset{S}{\ketbra{0}{0}}+ U_{11}^{(1)}   \underset{S}{\ketbra{1}{1}}\right)
+ \beta_1  \; U_{10}^{(1)} \underset{S}{\ketbra{1}{0}}
= \left[ \begin{array}{ll}
\beta_0 & 0 \\
\beta_1 U_{10}^{(1)}  & \beta_0 U_{11}^{(1)}\end{array} \right], \\ 
 \label{eq:Tilde(K)_0} 
V_{S}^{\dagger} K_{S}^{(0)}&=&
\left[ \begin{array}{ll}
V_{00}^*   & V_{10}^* \\
V_{01}^*  &V_{11}^*\end{array} \right]  \left[ \begin{array}{ll}
\beta_0 & 0 \\
\beta_1 U_{10}^{(1)}  & \beta_0 U_{11}^{(1)}\end{array} \right] 
= \left[ \begin{array}{ll}
\beta_0 V_{00}^* + \beta_1 V_{10}^* U_{10}^{(1)}   & \beta_0 V_{10}^* U_{11}^{(1)}\\
\beta_0 V_{01}^* + \beta_1 V_{11}^* U_{10}^{(1)} &\beta_0 V_{11}^*U_{11}^{(1)}\end{array} \right]  =:
 \left[ \begin{array}{ll}
\tilde{K}^{(0)}_{00}   & \tilde{K}^{(0)}_{01} \\
\tilde{K}^{(0)}_{10}  &\tilde{K}^{(0)}_{11}\end{array} \right] ,  
\end{eqnarray}
where  for later convenience we expressed the operators in matrix form w.r.t. the energy eigenbasis of the qubit Hamiltonian $H_S$. Invoking ~(\ref{defepsilonk}) we then get 
\begin{eqnarray} \label{FFAFA}
 \epsilon^{(0)}_C(\mathbf{\Phi}_{|\beta\rangle},\mathcal{V}) &=& \frac{1}{2}\left( \sum_{i,j} |\tilde{K}^{(0)}_{ij}|^2-\frac{1}{2} 
 |\sum_j\tilde{K}^{(0)}_{jj}|^2\right)\\
 &=& \frac{1}{2}\left( |\tilde{K}^{(0)}_{00}|^2+ |\tilde{K}^{(0)}_{11}|^2+ |\tilde{K}_{01}|^2+|\tilde{K}^{(0)}_{10}|^2-\frac{1}{2} | \tilde{K}^{(0)}_{00} + \tilde{K}^{(0)}_{11}|^2\right) \nonumber \\
  &=& \frac{1}{2}\left(\frac{1}{2} | \tilde{K}^{(0)}_{00}  - \tilde{K}^{(0)}_{11}|^2+ |\tilde{K}^{(0)}_{01}|^2+|\tilde{K}^{(0)}_{10}|^2\right) \nonumber \\
  &=&  \frac{1}{2}\left(\frac{1}{2}\big| \beta_0 (V_{00}^* -  V_{11}^*U_{11}^{(1)}) + \beta_1 V_{10}^* U_{10}^{(1)} \big|^2 
+ 
  |\beta_0 V_{10}^* U_{11}^{(1)}|^2 + |\beta_0 V_{01}^* + \beta_1 V_{11}^* U_{10}^{(1)}|^2\right). \nonumber \end{eqnarray} 
    Our goal is to show that, for all choices of $U^{(1)}$ 
    the above formula set a non trivial lower bound for 
$\epsilon_{C}(\mathbf{\Phi}_{|\beta\rangle},\mathcal{V})$ proportional to $|\beta_0|^2$ for all gates $V_S$ that do not commute with $H_S$.
Observe first that since $V_S$ is a $2\times2$ unitary operator,  up to an irrelevant global phase,
its matrix elements $V_{ij}$ can be parametrized as follows 
\begin{eqnarray} 
\begin{cases} \label{eq:paramV} 
&V_{00} = \cos \theta\;, \qquad  \quad V_{01}= \sin\theta e^{i\gamma}, \\
&V_{10} = \sin \theta e^{i(\gamma+\delta)} \;, \quad  V_{11} =- \cos\theta e^{i\delta}, 
\end{cases} 
\end{eqnarray} 
with some proper choices of $\theta\in [0,\pi/2]$ and $\gamma,\delta\in[0,2\pi[$. 
Furthemore  imposing the NEP condition is equivalent to require  its 
non-diagonal elements to have strictly positive non-diagonal contributions, i.e. 
\begin{eqnarray} 
 [V_{S},H_{S}]\neq 0 \quad \Longleftrightarrow \quad |V_{01}|=|V_{10}|=|\sin\theta| >0 \quad \Longleftrightarrow \quad
 \theta\neq 0, \pi/2. 
\end{eqnarray} 
We then notice that  the inequality~(\ref{thesis0})  can be easily verified  when $|V_{10}^* U_{11}^{(1)}|>0$. Indeed 
dropping two non-negative terms, from~(\ref{FFAFA}) we can write
\begin{eqnarray} \epsilon^{(0)}_C(\mathbf{\Phi}_{|\beta\rangle},\mathcal{V})\geq 
 |\beta_0 V_{10}^* U_{11}^{(1)}|^2/2 =
(\sin^2 \theta  |U_{11}^{(1)}|^2/{2}) |\beta_0|^2  \;,
\end{eqnarray} 
which allows us to identify the function $\eta(V_{S})$ of  Eq.~(\ref{thesis0}) with  $(\sin^2 \theta  |U_{11}^{(1)}|^2/{2}) $. 

Assume next  that $|V_{10}^* U_{11}^{(1)}|=0$ a condition which for NEPG $V_S$'s (i.e. for $\theta\neq 0,\pi/2$) is equivalent to have
$|U_{11}^{(1)}|=0$.
Under this assumption we can rewrite~(\ref{FFAFA}) as:  
\begin{eqnarray} 
\epsilon^{(0)}_C(\mathbf{\Phi}_{|\beta\rangle},\mathcal{V})
  &\geq&  \frac{1}{2}\left(\frac{1}{2}\big| \beta_0 V_{00}^*  + \beta_1 e^{i\alpha}  V_{10}^*  \big|^2 
+ |\beta_0 V_{01}^* + \beta_1e^{i\alpha}  V_{11}^*|^2\right) \nonumber \\
&=&  \frac{1}{2}\Big(\frac{1}{2} (|\beta_0 V_{00}^*|^2 + | \beta_1 V_{10}^*|^2 + 2 \Re[\beta^*_0 \beta_1 e^{i\alpha} V_{00}V_{10}^* ])\nonumber \\
&&+ |\beta_0 V_{01}^*|^2 + | \beta_1 V_{11}^*|^2 + 2 \Re[\beta^*_0 \beta_1 e^{i\alpha} V_{01}V_{11}^* ]\Big)\;,\label{eeee} 
\end{eqnarray} 
where we used the fact that since $U_{ij}^{(1)}$ is a $2\times 2$ unitary matrix,  $|U_{11}^{(1)}|=0$ implies $|U_{10}^{(1)}|=1$, i.e. $U_{10}^{(1)}=e^{i\alpha}$ for some real parameter $\alpha$.  Invoking hence the parametrization~(\ref{eq:paramV}) the above can then be  conveniently casted in the form
\begin{eqnarray} 
\epsilon^{(0)}_C(\mathbf{\Phi}_{|\beta\rangle},\mathcal{V})
 &\geq& \frac{1}{2} \Big( \frac{1}{2} (|\beta_0|^2 \cos^2\theta + | \beta_1|^2 \sin^2\theta+ 2\cos \theta \sin\theta\;  \Re[\beta^*_0 \beta_1 e^{i(\alpha-\gamma-\delta)}])\nonumber \\
&&+ |\beta_0 |^2 \sin^2\theta + | \beta_1|^2 \cos^2\theta  - 2 \cos \theta \sin\theta\; \Re[\beta^*_0 \beta_1 e^{i(\alpha+\gamma-\delta)} ]\Big)\nonumber \\
&=&   \frac{1}{2} \Big( \left( \tfrac{1+  \sin^2\theta}{2} \right)  |\beta_0|^2 + \left( \tfrac{1+  \cos^2\theta}{2}\right)   | \beta_1|^2 
+ \sin(2\theta) \Re\left[\beta^*_0 \beta_1 e^{i(\alpha-\delta)}\left(\frac{e^{-i\gamma}}{2}-e^{i\gamma}\right)\right]
\Big)\label{questaqui}  .
 \end{eqnarray} 
 Observe that if $\sin (2\theta)=0$ (i.e. for $\theta=\pi/4$), Eq.~(\ref{questaqui})  leads to \begin{eqnarray} 
\epsilon^{(0)}_C(\mathbf{\Phi}_{|\beta\rangle},\mathcal{V})
 \geq \frac{3}{8}  ( |\beta_0|^2 +    | \beta_1|^2 )  \geq  \frac{3}{8}  |\beta_0|^2, 
 \end{eqnarray} 
 which proves the thesis setting $\eta(V_S) =3/8$. To deal with the case  $\theta\neq  \pi/4$ instead we observe that the r.h.s. of Eq.~(\ref{questaqui}) can be lower bounded as follows
 \begin{eqnarray} 
\epsilon^{(0)}_C(\mathbf{\Phi}_{|\beta\rangle},\mathcal{V}) &\geq&  \frac{1}{2} \Big(  \left( \tfrac{1+  \sin^2\theta}{2} \right)  |\beta_0|^2 + \left( \tfrac{1+  \cos^2\theta}{2}\right)   | \beta_1|^2 
 -|\sin(2\theta)\beta_0 \beta_1|\left|\frac{e^{-i\gamma}}{2}-e^{i\gamma}\right|\Big) \nonumber \\
  &=&  \frac{1}{2} \Big(  \left( \tfrac{1+  \sin^2\theta}{2} \right)  |\beta_0|^2 + \left( \tfrac{1+  \cos^2\theta}{2}\right)   | \beta_1|^2 
 -|\sin(2\theta)|\sqrt{5/4-\cos(2\gamma)}|\beta_0 \beta_1|\Big). \label{questaqua} 
\end{eqnarray} 
Now  from the normalization condition of $|\beta\rangle_B$, it follows that we can write $|\beta_1|= x \sqrt{1-|\beta_0|^2}$ with $x\in[0,1]$. Accordingly we can rewrite (\ref{questaqua}) as
\begin{eqnarray} 
\epsilon^{(0)}_C(\mathbf{\Phi}_{|\beta\rangle},\mathcal{V}) &\geq&    A   x^2 
 -B x+C\;, \nonumber
\end{eqnarray} 
with \begin{eqnarray}
\begin{cases} &A: =  \left({1+  \cos^2\theta}\right)  (1-|\beta_0|^2)/4, \\
&B:= |\sin(2\theta)|\sqrt{5/4-\cos(2\gamma)}|\beta_0| \sqrt{1-|\beta_0|^2}/2, \\ 
&C:= \left( 1+  \sin^2\theta  \right)  |\beta_0|^2/4. \end{cases}\end{eqnarray} 
As a function of $x$ the above expression is a parabola with reaches the minimum value $C-B^2/(4A)$  for 
 $x_{\min} = B/2A$, so that 
 \begin{eqnarray} 
\epsilon^{(0)}_C(\mathbf{\Phi}_{|\beta\rangle},\mathcal{V})
   &\geq& C-B^2/(4A)
 =\left(  \left( \frac{1+  \sin^2\theta}{4} \right)- \frac{\sin^2(2\theta)|5/4-\cos(2\gamma)| }{4\left( {1+  \cos^2\theta}\right) }
 \right) |\beta_0|^2\nonumber \\ 
 &=&\frac{ (1+  \sin^2\theta)(1+  \cos^2\theta)- \sin^2\theta \cos^2\theta |5-4\cos(2\gamma)| }{4 \left( {1+  \cos^2\theta}\right) }
 |\beta_0|^2\nonumber \\ 
 &\geq &\frac{ (1+  \sin^2\theta)(1+  \cos^2\theta)- 9\sin^2\theta \cos^2\theta  }{4 \left( {1+  \cos^2\theta}\right) }
 |\beta_0|^2\nonumber \\ 
 &= &\frac{ 1  - 4\sin^2\theta \cos^2\theta  }{2({1+  \cos^2\theta}) }=\frac{ (\cos^2\theta -\sin^2\theta)^2  }{2({1+  \cos^2\theta}) }
 |\beta_0|^2,
\end{eqnarray}
which proves the thesis with $\eta(V_S)=\frac{ (\cos^2\theta -\sin^2\theta)^2  }{2({1+  \cos^2\theta}) }$ that is strictly positive function for $\theta\neq \pi/4$. 
\end{proof}
  We remark that this is a property that applies to the population of the ground state given its special role. Indeed, if theorem \ref{Th:BoundaryOcc} was true for every $\beta_{n}$, i.e. $\epsilon^{(n)}_C(\mathbf{\Phi}_{|\beta\rangle},\mathcal{V})\geq \eta(V_{S})\abs{\beta_{n}}^{2}$, we would have $\epsilon_C(\mathbf{\Phi}_{|\beta\rangle},\mathcal{V})=\sumab{n=0}{\infty} \epsilon^{(n)}_C(\mathbf{\Phi}_{|\beta\rangle},\mathcal{V})\geq \eta(V_{S})$, meaning that arbitrary small precision would be impossible. 

The theorem above implies that \lq\lq smooth" states as in Eq.~(10) 
of the main text need to have zero overlap with the ground state in order to reach the optimal asymptotic scaling of the Choi infidelity in $\delta$, as we prove below.
\begin{cor} \label{Cor:GSpopulation}
    Let $\ket{\beta^{(\delta)}}:= C_{\delta}\underset{n=0}{\overset{\infty}{\sum}}\psi(n\delta)\ket{n}_{B} $ as in 
    Eq.~(10), and let $\mathbf{\Phi}_{\ket{\beta^{(\delta)}}}$ be the respective channel as in Eq.~(1) of the main text
    \begin{equation}
        \epsilon_{C}(\mathbf{\Phi}_{|\beta^{(\delta)}\rangle},\mathcal{V})\geq \epsilon^{(0)}_C(\mathbf{\Phi}_{|\beta\rangle},\mathcal{V}) \geq \eta(V_S) \abs{\psi(0)}^{2}\delta+o(\delta).
    \end{equation}
\end{cor}
\begin{proof}
    Thanks to theorem \ref{Th:BoundaryOcc} we have $\epsilon_{C}^{(0)}(\mathbf{\Phi}_{|\beta^{(\delta)}\rangle},\mathcal{V})\geq \eta(V_S) \abs{\psi(0)}^{2}C_{\delta}^{2}+o(\delta)$, then using $C_{\delta}^{2}=\delta +o(\delta)$ we get the thesis.
\end{proof}
To conclude the argument, notice that Eq.~(11) in the main text 
implies that all the smooth battery states with zero overlap with the ground state (equivalently, with $\psi(0)=0$), allows for $
\epsilon_{C}(\mathbf{\Phi}_{|\beta^{(\delta)}\rangle},\mathcal{V}) \propto \delta^{2}
$.

More generally  suppose there exists a unitary interaction 
$U_{SB}$ and a battery state $|\beta'\rangle_B$ (possibly different from 
the $|\beta^{(\delta)}\rangle_B$ reported in the main text) 
that 
 has zero overlap with the ground state $|0\rangle_B$ 
 (i.e. ${_B \langle} 0| \beta'\rangle_B=0$), such that
   \begin{equation}\label{performance} 
    \epsilon_{C}(\mathbf{\Phi}_{|\beta'\rangle},\mathcal{V})\leq C(V_S)\delta^{2},
\end{equation}
for some appropriate function $C(V_S)$. According to Theorem~\ref{Th:BoundaryOcc}, 
any other pure battery state $|\beta\rangle_B$ with
non-zero ground state amplitude $\beta_0$ 
can achieve the same precision level~(\ref{performance}) only if the following inequality holds:
\begin{equation} \label{Eq:beta0bound}
\eta(V_S) |\beta_0|^2 \leq C(V_S)\delta^2   \quad \Longrightarrow \quad |\beta_0|^{2} \leq \delta^{2} \frac{C(V_S)}{\eta(V_S)}.
\end{equation}
Importantly, this conclusion remains valid even if a different $U_{SB}$ -- distinct from the one used to establish inequality (\ref{performance}) -- is employed in the model. 

We conclude by noticing that the choice $\beta_0 = 0$ allows for a simplification in Eq. (6). Indeed, we 
have
\begin{equation}
    \Delta(0,\beta_1,V_S) = |V_{01}|^2\left( 1 - \frac{|V_{01}|^2}{4}\right) |\beta_1|^2.
\end{equation}
For the smooth family defined above (see Eq.~(10) 
in the main text) we have
$\beta_1 = \beta_0 + \delta C_{\delta} \psi'(0) + O(\delta^2) $
 so we get 
\begin{equation}
\Delta(0,\beta_{1},V_{S}) \sim \abs{\beta_{1}}^{2}=C_{\delta}^{2}\delta^{2}\abs{\psi'(0)}^{2}+O(\delta^{4})=
O(\delta^{3}),
\end{equation}
where we used that $C_{\delta}^{2}=\delta+o(\delta)$.
Given this, the contribution of $\Delta(0, \beta_1, V_S)$ can be safely neglected.

\subsection{On the optimality of the unitary interaction between S and B}\label{subsec:ChosenInteraction}
The aim of this subsection is to justify the system-battery interaction chosen in the main text, 
\begin{align} \label{eq:usbUPapp}
& U_{SB}  = \underset{S}{\ketbra{0}{0}} \otimes \underset{B}{\ketbra{0}{0}}
 +\underset{n>0}{\overset{\infty}{\bigoplus}}\; {U}^{(n)}_{SB}, \quad  {U}^{(n)}_{SB} := \sum_{i,j=0,1} \underset{B}{\ketbra{n-i}{n-j}} \otimes \underset{S}{\ketbra{i}{j}}V_{ij}. 
\end{align}
 We first give an intuitive argument and then proceed with a rigorous proof of its optimality under an additional assumption.
Recall that in the case of the qubit-battery model, 
the most general unitary interaction $U_{SB}$ that is energy preserving, i.e. $[U_{SB},H_{S}+H_{B}]=0$, 
is the one in Eq.\eqref{eq:usbUPgen}.
Assuming as input state of the battery one of the 
states~$|\beta^{(\delta)}\rangle_B$ of the ansatz 
in Eq.~(10),
the  Kraus operators of the associated map 
 $ \mathbf{\Phi_{|\beta^{(\delta)}\rangle}}$ are (for $n>0$) 
\begin{equation}
    K_{S}^{(n,\delta)}= {_B \langle} n|  U_{SB} |\beta^{(\delta)}\rangle_B= \underset{i,j=0,1}{\sum} \beta^{(\delta)}_{n+i-j} U^{(n+i)}_{ij} 
\underset{S}{\ketbra{i}{j}},
\end{equation}
from which, after some manipulations, we obtain
\begin{align} \label{eq:bla}
    K_{S}^{(n,\delta)} &= \beta_{n}^{(\delta)}U_S^{(n)}+\underset{i,j=0,1}{\sum} ( \beta_{n+i-j}^{(\delta)} U^{(n)}_{ij}-\beta_{n}^{(\delta)} U^{(n+i)}_{ij}) 
\underset{S}{\ketbra{i}{j}} \\ &= 
\beta_{n}^{(\delta)}U_S^{(n)}+\underset{i,j=0,1}{\sum} \big[ (\beta_{n+i-j}^{(\delta)} -\beta_n^{(\delta)})  U^{(n)}_{ij}-\beta_{n}^{(\delta)} (U_{ij}^{(n+i)} - U^{(n)}_{ij}) \big]
\underset{S}{\ketbra{i}{j}}.
\end{align}
 Unitary channels are extreme points of the convex set of CPTP maps, thus any Kraus representation of a unitary channel is made of operators proportional to the unitary itself, so we want to choose $U^{(n)}$ such that the convergence
$K_{S}^{(n,\delta)}\overset{\delta \rightarrow 0}{\propto} V_{S}$, is as quick as possible. Considering a smooth battery state as in the main text,
the first term inside the sum in
Eq. \eqref{eq:bla} becomes of order $\delta$ since $\frac{\beta_{n+i-j}^{(\delta)}-\beta_{n}^{(\delta)}}{\beta_{n}^{(\delta)}}= O(\delta)$.
To make the rest of the terms proportional to $V_S$, it is sufficient to choose all the $U^{(n)}$ to be the same and equal to $V_S$.
In this way, the term $\beta_n U^{(n)}_S$ is exactly proportional to $V_S$ (the condition we were looking for) while the last term in the sum nullifies.

We now come to the rigorous statement. Suppose that $\abs{U_{ij}^{(n)}}=\abs{U_{ij}}\;\; \forall \; n$, then the interaction can always be parametrized as 
\begin{eqnarray} \label{eq:paramU} 
\begin{cases} 
&U_{00}^{(n)} = \cos \bar{\theta} e^{i\alpha_{n}}\;, \qquad \quad \quad U_{01}^{(n)}= \sin \bar{\theta} e^{i(\alpha_n+\gamma_n)}\;, \\
&U_{10}^{(n)} = \sin \bar{\theta} e^{i(\alpha_{n}+ \gamma_{n}+\delta_{n})} \;, \quad  U_{11}^{(n)} =- \cos\bar{\theta} e^{i(\alpha_{n}+\delta_{n})}\;,
\end{cases} 
\end{eqnarray}
for some arbitrary phases $\alpha_{n},\delta_{n},\gamma_{n}$, and a fixed, uniform over $n$, angle $\bar{\theta}$. Then we have the following bound on the Choi-infidelity
\begin{prop}
Suppose that $\beta_{0}=0$ and consider the parametrizations of Eq.\eqref{eq:paramV} for $V_{S}$ and Eq.\eqref{eq:paramU} for $U_{SB}$, then 
\begin{equation} \label{dfdsf} 
    \epsilon_{C}(\mathbf{\Phi}_{\beta},\mathcal{V}_S) \geq 
 \sin^2(\bar{\theta}-{\theta}) +\frac{1}{4}\Big[ \sin (2 \bar{\theta}) \sin (2\theta )  Q_{\beta}^{(1)}  
  +   \sin^2 \bar{\theta} \sin^2\theta \; Q_{\beta}^{(2)} \Big] \;,
\end{equation}
where $Q_{\beta}^{(1)} :=2- 2\sumab{n=1}{\infty}\abs{\beta_{n}\beta_{n+1}} 
  \;$, and $Q_{\beta}^{(2)} := 2\left(1-   \sumab{n=1}{\infty} 
  |\beta_n \beta_{n+2} |\right)$
  . 
\end{prop}
Before presenting the proof, we underline the consequences of Eq. \eqref{dfdsf}. Since the RHS of \eqref{dfdsf} is lower bounded by $\sin^{2}(\bar{\theta}-\theta)$ if one wants to fix the interaction irrespectively of the battery state, the choice $\bar{\theta}=\theta$, corresponding to the interaction assumed in the main text in Eq.(5), is the only one allowing for arbitrary small Choi-infidelity. Furthermore, choosing the optimal angle as a function of the battery state only produces a negligible enhancement with respect to setting $\bar{\theta}=\theta$ regardless of $\beta$. More formally let $G(\theta,\bar{\theta},Q^{(1)}_{\beta},Q^{(2)}_{\beta})$ be the RHS of Eq.\eqref{dfdsf}, and $G_{min}(Q^{(1)}_{\beta},Q^{(2)}_{\beta}):= \underset{\bar{\theta}}{\min} \; G(\theta,\bar{\theta},Q^{(1)}_{\beta},Q^{(2)}_{\beta})$ then using that $\partial_{\bar{\theta}}G(\theta,\bar{\theta},Q^{(1)}_{\beta},Q^{(2)}_{\beta})_{|\bar{\theta}=\theta}=O(G(\theta,\theta,Q^{(1)}_{\beta},Q^{(2)}_{\beta})) $ we have 
\begin{equation}
G(\theta,\theta,Q^{(1)}_{\beta},Q^{(2)}_{\beta})= G_{min}(Q^{(1)}_{\beta},Q^{(2)}_{\beta}) +O(G(\theta,\theta,Q^{(1)}_{\beta},Q^{(2)}_{\beta})^{2}).
\end{equation}  Thus, in conclusion, the choice $\bar{\theta}=\theta$ is asymptotically optimal, leading to the the following lower bound to the Choi infidelity: 
\begin{align}
  &  \epsilon_{C}(\mathbf{\Phi}_{\beta},\mathcal{V}_S) \geq \frac{1}{4}\left[\sin^2 (2\theta ) \; 
Q_{\beta}^{(1)}    +\sin^4\theta \;  Q_{\beta}^{(2)}\right]+O(\epsilon_{C}^{3/2})= \sin^2\theta \;  Q_{\beta}^{(1)}   
- \sin^4 \theta \left( Q_{\beta}^{(1)}  - \frac{Q_{\beta}^{(2)}}{4}\right)+O(\epsilon_{C}^{3/2}) .
\end{align}
For real coefficients $\beta_{n}$ we have, using summation by parts, $Q_{\beta}^{(1)}=2\sumab{n=0}{\infty}\abs{\beta_{n+1}-\beta_{n}}$ which implies  
\begin{align}
&=|V_{01}|^2 \sum_{n=1}^{\infty} \big[  |\beta_{n+1}  - \beta_n|^2  - \frac{|V_{01}|^2}{4}   |\beta_{n+1} + \beta_{n-1} - 2 \beta_n|^2 \big] + \left(\abs{V_{01}}^{2}-\frac{\abs{V_{01}}^{4}}{4}\right)\abs{\beta_{1}}^{2}+O(\epsilon_{C}^{3/2})
\end{align}
recovering the expression we have in Eq.(6), and proving the optimality of the interaction \eqref{eq:usbUP}, under the assumption of an angle $\bar{\theta}$ uniform over $n$.
\begin{proof}
We start with the observation that we can, without loss of generality, always choose each $\beta_{n}$ coefficient to be real and positive. Indeed, their phases can be written as the application of an energy preserving gate on $B$ and as such they can be reabsorbed into the interaction $U_{SB}$. 
Then we rewrite the Choi-infidelity in a compact form, letting
\begin{eqnarray} 
\kappa_n(\beta)&:=& \frac{\mbox{Tr} [ V^\dag_S \langle n| U_{SB} |\beta\rangle] }{d}, \end{eqnarray} 
the corresponding  Choi error is then given by
\begin{eqnarray} 
 \epsilon_{C}(\mathbf{\Phi}_{\beta},\mathcal{V}_S)&=&1-\sumab{n=0}{\infty}  \left| \kappa_n(\beta)\right|^{2}= 
1 -\frac{1}{d^2}\langle \beta  |\left( \sum_{j,j'}  \langle j| U^\dag_{SB}V_S |j\rangle\langle j'|  V^\dag_S  U_{SB} |j'\rangle\right) |\beta\rangle \nonumber \\
&=& 1- \langle \beta | M |\beta\rangle, \nonumber 
 \end{eqnarray} 
 with
 \begin{eqnarray} 
 M &:=& \frac{1}{d^2} \sum_{j,j'}  \langle j| U^\dag_{SB}V_S |j\rangle\langle j'|  V^\dag_S U_{SB} |j'\rangle
 = Q^\dag Q \;, \nonumber \\
Q &:=&\frac{1}{d} \mbox{Tr}_S [ V^\dag_S U_{SB}] \nonumber. \label{defdiqu} 
 \end{eqnarray} 
 For $n>0$ we have
 \begin{equation} 
 Q|n\rangle = \frac{1}{2} \left[  U_{10}^{(n)} V_{10}^* |n-1\rangle+ 
 \label{qn}
 (U_{00}^{(n)}  V_{00}^* + U_{11}^{(n+1)} V_{11}^* ) | n\rangle + U_{01}^{(n+1)} V_{01}^* |n+1\rangle\right]\;.
 \end{equation}
 Therefore, invoking the parametrization \eqref{eq:paramU} we can write 
 \begin{eqnarray} 
  \langle n| Q^\dag Q  |n\rangle &=&  \frac{1}{4} \left[  \nonumber 
 |U_{00}^{(n)} - e^{-i\delta} U_{11}^{(n+1)}|^2  \cos^2\theta  + (|U_{10}^{(n)} |^2+ |U_{01}^{(n+1)}|^2)  \sin^2\theta \right]\;, \\
  \langle n| Q^\dag Q  |n+1\rangle &=&  \frac{\cos\theta \sin\theta}{4} \Big[  \nonumber 
  (U_{00}^{(n)}  -   e^{-i\delta} U_{11}^{(n+1)}  )^* U_{10}^{(n+1)}  e^{-i(\gamma+\delta)} \nonumber \\
  &&+U_{01}^{(n+1)*} e^{i\gamma}  (U_{00}^{(n+1)}   -  e^{-i\delta}U_{11}^{(n+2)} )
  \Big]\;, \\
   \langle n| Q^\dag Q  |n-1\rangle &=&\left( \langle n-1| Q^\dag Q  |n\rangle \right)^*\;, \\
   \langle n| Q^\dag Q  |n+2\rangle &=&  \frac{ \sin^2\theta}{4}
   U_{01}^{(n+1)*}U_{10}^{(n+2)} e^{-i\delta} \;, \\
     \langle n| Q^\dag Q  |n-2\rangle &=& \left( \langle n-2| Q^\dag Q  |n\rangle \right)^*,
 \end{eqnarray} 
 while all the other products nullify. Now observe that 
 \begin{align} 
 &\langle \beta|M |\beta\rangle =  \sum_{n,n'} \beta_{n} \beta_{n'} \langle n | Q^\dag Q  |n'\rangle \nonumber \\
&  \leq \sum_{n=0}^{\infty}  \big(  |\beta_n|^2 \langle n | Q^\dag Q  |n\rangle 
  + 2 \beta_n \beta_{n+1}  \abs{\langle n | Q^{\dag} Q  |n+1\rangle}
  +  2 \beta_n \beta_{n+2} \abs{ \langle n | Q^{\dag} Q  |n+2\rangle} \big)  \;. 
\end{align} 
In order to bound the last contribution, we can employ the parametrization
\eqref{eq:paramU} and obtain
 \begin{eqnarray} 
 &&\beta_n \beta_{n+2}  \abs{\langle n | Q^\dag Q  |n+2\rangle}= \frac{ \sin^2\theta  }{4} \beta_n \beta_{n+2}  \sin^{2} \bar{\theta}  \;.  \nonumber 
 \end{eqnarray} 
  This allows us to establish the following upper bound 
  \begin{eqnarray} 
 \langle \beta|M |\beta\rangle &\leq& \sum_{n=0}^{\infty}    |\beta_n|^2 \abs{\langle n | Q^\dag Q  |n\rangle} 
  + 2  |\beta_n \beta_{n+1} | \left| \langle n | Q^\dag Q  |n+1\rangle\right|  
 +  2 \frac{ \sin^2\theta  }{4}\ \beta_n \beta_{n+2} \sin^{2} \bar{\theta}   \;.
 \end{eqnarray} 
Observe next that 
 \begin{align}
& |\langle n| Q^\dag Q  |n+1\rangle| \leq   
   \cos\theta \sin\theta  \cos \bar{\theta}\sin\bar{\theta}  \;, \\
& |\langle n| Q^\dag Q  |n\rangle|\leq \cos^{2}\bar{\theta} \cos^2\theta  +\sin^2 \bar{\theta}\sin^2\theta .
 \end{align} 
  Accordingly we can claim that 
 \begin{equation}
\langle \beta|M |\beta\rangle \leq \cos^2(\bar{\theta}-{\theta})   -\frac{1}{4}\left[\sin (2 \bar{\theta}) \sin (2\theta )   \left(1-   \sumab{n=1}{\infty} 
 \beta_n \beta_{n+1} \right) + \sin^2 \bar{\theta} \sin^2\theta \left(1-\sum_{n=1}^{\infty}  \beta_n \beta_{n+2}\right)\right].
 \end{equation} 
We get the thesis by recalling that $\epsilon_{C}(\mathbf{\Phi}_{\beta},\mathcal{V}_S)=1- \langle \beta|M |\beta\rangle $.
  \end{proof}

\section{Minimizing the Choi infidelity using $N$ levels via  discrete sine transform}
\label{sec:fourier}
In this section we solve explicitly the optimization of the Choi infidelity for a battery with a fixed number of levels $N$. Since we know that the solution should have vanishing boundary conditions (in light of App. \ref{sec:calcchoi2} and the requirement that the levels above the $N$-th are not populated), we choose the family $\tilde{\beta}_k$ such that

\begin{equation} \label{eq:defFou}
    \beta_n = \sqrt{\frac{2}{N}}  \sum_{k=1}^{N-1} \sin\left( \frac{\pi k n}{N} \right) \tilde{\beta}_k,
\end{equation}

that is, a discrete sine transform of the original coefficients.
Since we have 

\begin{equation}    \sum_{n=1}^{N-1} \sin \left(\frac{\pi k n}{N}\right) \sin \left(\frac{\pi k'n}{N} \right) =  \frac{N}{2} \delta_{k,k'}
\end{equation}

the change of basis defined above is orthogonal and preserves normalization:

\begin{equation}
    \sum_{n=1}^{N-1} |\beta_n|^2 = \sum_{k=1}^{N-1} |\tilde{\beta}_k|^2 = 1.
\end{equation}

We notice that the first term of the Choi infidelity is nothing but the discrete Laplacian with Dirichlet boundary conditions \cite{DiscreateGreenFunctions}.
The discrete sine transform diagonalizes the discrete Laplacian, so that we have

\begin{equation}
    \sum_{n=0}^{N-1}|\beta_{n+1} - \beta_{n}|^2
    = \sum_{k=1}^{N-1}   4 \sin^2\left( \frac{\pi k}{2N} \right) | \tilde{\beta}_k |^2.
\end{equation}

In the same way, the discrete second derivative can be expressed as

\begin{equation}
   \sum_{n=0}^{N-1}|\beta_{n+1} + \beta_{n-1} - 2\beta_{n}|^2
   = \sum_{k=1}^{N-1} 16 \sin^4 \left( \frac{\pi k}{2N} \right) | \tilde{\beta}_k |^2.
\end{equation}

We conclude that 
\begin{equation}
\epsilon_C = 4 |V_{01}|^2  \sum_{k=0}^{N-1} \sin \left( \frac{\pi k}{2N} \right)^2 \big[1 -  |V_{01}|^2\sin \left( \frac{\pi k}{2N} \right)^2 \big] 
|\tilde{\beta}_k|^2 .   \label{eq:choifour}
\end{equation}
The latter is of the form $\epsilon_c=\abs{V_{01}}^{2}\sumab{k=1}{N-1}f(k,V_{s})\abs{\tilde{\beta}_{k}}^2$ where $f(k,V_{S}):=4\sin \big( \frac{\pi k}{2N} \big)^2 \big[1 - |V_{01}|^2 \sin \big( \frac{\pi k}{2N} \big)^2 \big]$ and thus its minimization is equivalent to the minimization of $f(k,V_{S})$. \\
For any $V_{S}$ the absolute minimum is reached in $k=1$ \footnote{If $\abs{V_{01}}^{2}=1$ there is a second equivalent minimum in $k=N-1$, however the corresponding state is harder to prepare and suboptimal for any other $V_{S}$, hence holds no advantages over the $k=1$ solution}, where we get $f(1,V_{S})=4 \sin(\frac{\pi}{2N})^{2}\big[1 - |V_{01}|^2 \sin \big( \frac{\pi}{2N} \big)^2 \big]=\frac{\pi^{2}}{N^{2}} -O(N^{-4})$, where in the second equality we considered the $N\rightarrow \infty$ limit. In this way we proved that the optimal state in the sine-shaped one in Eq.\eqref{Eq:OptNstate} and the inequality \begin{equation}
    N \geq \frac{ \pi |V_{01}|}{\sqrt{\epsilon_C}}+o(1).
\end{equation}

\section{Explicit optimization of the precision at fixed battery resources}
\label{sec:resoptim}
In this section, we explicitly solve the Lagrangian optimization problems presented in the main text for the case of energy, squared energy and Quantum Fisher Information. 

\subsection{Minimizing the Choi infidelity at fixed $\langle E^{2}\rangle$}
We solve explicitly the optimizaion problem when the resource considered is the average squared energy, $\mathcal{R}(\beta):=\mbox{Tr}[H_{B}^{2}\beta]$. 
Thus we start from 
\begin{align} \label{eq:ground1Dapp}
\mathcal{S}_{\mathcal{R}}\hspace{-1mm}=\hspace{-1mm} \int_0^{\infty} \hspace{-2mm}\dd x\, \psi(x) \left[  -  \delta^{2} \frac{\dd^2}{\dd x^2}  +\lambda \left(\frac{x \omega}{\delta} \right)^{m_{\mathcal{R}}} \right] \psi^*(x),
\end{align}
with $m_{\mathcal{R}}=2$, i.e.  
\begin{align} \label{eq:groundE^{2}app}
\mathcal{S}_{\langle E^2 \rangle}[\psi(x)] = \delta^{2} \int_0^{\infty} \dd x\, \psi(x) \big( x^{2}a  - \frac{\dd^2}{\dd x^2}  \big) \psi^*(x),
\end{align}
where we defined $a:=\frac{\lambda \omega^{2}}{\delta^{4}}$. The variational problem defined by the functional in \eqref{eq:groundE^{2}app} is equivalent to finding the wave function of the ground state a quantum Harmonic oscillator constrained on $[0,\infty)$, and vanishing boundary condition $\psi(0)=0$. Such wave function is thus proportional to the lowest-energy odd wavefunction of a standard Harmonic oscillator of same frequency, i.e. 
\begin{equation}
    \psi_a(x)=\sqrt{2}a^{1/8}\psi_{1}(a^{1/4}x) \quad x\geq 0,
\end{equation}
where $\psi_{1}(x)$ is the normalized first Hermite function.  
We have now to compute $\langle E^{2}\rangle(a):=\omega^{2}\int_{0}^{\infty} \abs{\psi_{a}(x)}^{2}x^{2} \dd x\,$ and ${\rm UD}(a):=\int_{0}^{\infty}\abs{\psi'_{a}(x)}^{2} \dd x$. Both computations can be achieved through the virial theorem for the Harmonic oscillator, thaks to which we find  
\begin{align}
     \langle E^{2}\rangle(a)= \frac{3\omega^{2}}{2\sqrt{a}},
    \qquad {\rm UD}(\psi_{a})= \frac{3}{2}\sqrt{a}.
\end{align}
Inverting the first equation and using the shorthand notation ${\rm UD}_{min}(\langle E^{2}\rangle):={\rm UD}(\psi_{\langle E^{2}\rangle})$ we conclude
\begin{align}
    & \psi_{\langle E^{2}\rangle}= \sqrt{2}\left(\frac{3\omega^{2}}{2\langle E^{2}\rangle}\right)^{1/4} \psi_{1}\left(\sqrt{\frac{3\omega^{2}}{2\langle E^{2}\rangle}} x\right) , \\
    & {\rm UD}_{min}(\langle E^{2}\rangle)= \frac{9}{4}\frac{\omega^{2}}{\langle E^{2}\rangle}.
\end{align}
  
For comparison, for the sine-shaped state minimizing the number of levels
\begin{equation} \label{Eq:OptNstateApp}
\ket{\beta^{(N)}}_B:=C_{N}\underset{n=1}{\overset{N-1}{\sum}}\sin(\frac{\pi n}{N} )\ket{n}_B, 
\end{equation} 
we have 
\begin{equation}
{\rm UD}(\langle E^{2}\rangle)=\left(\frac{\pi^{2}}{3}-\frac{1}{2}\right)\frac{\omega^{2}}{\langle E^{2}\rangle}\approx 2.80 \frac{\omega^{2}}{\langle E^{2}\rangle}.
\end{equation}

\subsection{Minimizing the Choi infidelity at fixed mean energy}
We solve explicitly the optimization problem when the resource considered is the average energy, $\mathcal{R}(\beta):=\mbox{Tr}[H_{B}\beta]$. 
Thus we start from Eq.\eqref{eq:ground1D} with $m_{\mathcal{R}}=1$: 
\begin{align} \label{eq:groundE^{2}}
\mathcal{S}_{\langle E \rangle}[\psi(x)] = \delta^{2} \int_0^{\infty} \dd x\, \psi(x) \left( xa  - \frac{\dd^2}{\dd x^2}  \right) \psi^*(x),
\end{align}
where this time $a:=\frac{\lambda \omega}{\delta^{3}}$.
We have to solve the Schrödinger equation:
\begin{equation} \label{eq:dedimsc}
    -\psi''(x)+ax\psi=E\psi.
\end{equation}
for the lowest possible $E$ and with boundary condition $\psi(0)=0$.
It's convenient to introduce a new function $\phi$ and a new variable $y$ such that $\psi_{a}(x):=\phi(a^{1/3}x-\frac{E}{a^{2/3}}):=\phi(y) $.
In terms of $\phi$ the Eq.\eqref{eq:dedimsc} becomes
\begin{equation}
    \phi''(y)=y\phi(y).
\end{equation}
The solutions of this equation are well known and are given by linear combinations of Airy functions $\phi(y)=CAi(y)+C'Bi(y)$, however,
the normalization condition imposes $C'=0$. 
The smallest value of $Ea^{-2/3}$ for which $\psi_{a}(0)=CAi(-\frac{E}{a^{2/3}})=0$, is satisfied is known to be $x_{0}\approx 2.338$, thus, let us set $E_{g}:= x_{0}a^{2/3}$ and 
\begin{equation}
    \psi_{a}(x):=\phi(a^{1/3}x-E_{g}a^{-2/3})=\sqrt{\overline{C}a^{1/3}}Ai\left(xa^{1/3} -x_{0}\right), 
\end{equation}
where we defined  $\overline{C}:=\left[ \int_{0}^{\infty}\abs{Ai\left(x-x_{0}\right)}^2 \dd x \right]^{-1}$.\\
We complete the calculation computing $\langle E\rangle(a):=\omega \int_{0}^{\infty} \abs{\psi_{a}(x)}^{2}x \dd x $ and ${\rm UD}(\psi_{a})$:
\begin{align}
    &\langle E\rangle(a)= \omega \overline{C} C_{1}a^{-1/3} \quad  C_{1}:=\int_{0}^{\infty}x\abs{Ai\left(x-x_{0}\right)}^2 \dd x,\, \\
    &{\rm UD}(\psi_{a})=a^{2/3}\overline{C}C_{2} \quad C_{2}:=\int_{0}^{\infty}\abs{\partial_{x}Ai\left(x-x_{0}\right)}^2\dd x,\,
\end{align}
yielding 
\begin{align}
    & \psi_{\langle E \rangle}(x)=C_{E} Ai\left(x\frac{\omega \overline{C}C_{1}}{\langle E\rangle}-x_{0}\right), \\
    & {\rm UD}_{min}(\langle E\rangle)= \overline{C}^{3}C_{1}^{2}C_{2}\frac{\omega^{2}}{\langle E\rangle^{2}}:=\eta^{2}\frac{\omega^{2}}{\langle E\rangle^{2}} \approx 1.888 \frac{\omega^{2}}{\langle E\rangle^{2}}.
\end{align}
For completeness, we report the approximated numerical values of the integrals:
$\overline{C}\approx 2.033; \;\;C_{1}\approx 0.766 ; \;\; C_{2} \approx 0.383$.

For comparison, for the sine-shaped state in Eq.\eqref{Eq:OptNstate} we have 
\begin{equation}
     {\rm UD}(\langle E\rangle)= \frac{\pi^{2}}{4}\frac{\omega^{2}}{\langle E\rangle^{2}} \approx 2.466 \frac{\omega^{2}}{\langle E\rangle^{2}}.
\end{equation}

\subsection{Minimizing the Choi infidelity at fixed QFI} 
\label{sec:QFI}
 The Quantum Fisher information plays an important role in different fields of quantum information science, including quantum metrology, quantum speed limits, quantum thermodynamics and recently resource theory of asymmetry \cite{GiovaMetrology(2011),GiovaSpeed,Q.SpeedReview,DinamicalQFI, QFIBound, QFI-NonAbelian}. It can be defined in great generality, whenever we have a continuous parametrization of density matrices $\rho_{t}$ \cite{QFIDef, BuresQFI} as 
 \begin{align}
     &QFI(\rho_t)=2\partial_{t'}^{2}D_{B}^{2}(\rho_{t'},\rho_{t})_{|t'=t} 
\end{align}
 where $D_{B}(\rho,\sigma):=\sqrt{2(1-\sqrt{F(\rho,\sigma)})}$ is the Bures metric. 
 Here we only consider the case of a Hamiltonian parametrization, i.e. $\rho_t:=\exp[-iHt] \rho \exp[iHt] $, where the quantum Fisher information with respect to parameter $t$, only depends on $\rho $ and $H$, 
 \begin{equation}
     \mathcal{F}(\rho,H):= QFI(\exp[-iHt]\rho \exp[iHt]).
 \end{equation}
 The functional $\mathcal{F}(\rho,H)$ can be expressed in different ways, we recall that it
  equals 4 times the convex roof of the variance \cite{QFI-ConvexRoof}:
\begin{equation} \label{QFIConvexRoofExpression}
  \mathcal{F}(\rho,H)=4\underset{\{p_j,\ket{j}\}_{j} \in \mathbb{E}[\rho]}{\min} \sum_{j} p_j Var[\ket{j},H]= \mbox{Tr}[H^2 \rho]-\underset{\{p_j,\ket{j}\}_{j} \in \mathbb{E}[\rho]}{\max} \sum_{j} p_j\bra{j}H\ket{j}^2,
\end{equation}
where $\mathbb{E}[\rho]:= \bigg\{ \{ p_{j},\ket{j}\}_{j}: \underset{j}{\sum} p_j \ketbra{j}{j} = \rho\bigg\}$ is the set of all ensemble representations of $\rho$.

In \cite{QFIBound} the minimum amount of QFI the battery state must possess to perform an $\epsilon_{wc}$ precise unitary is found. Their construction holds for a generic Hamiltonian  $H_{S}$ and is extremely complex, the battery system is divided into 5 different sub-systems and the proof involves many lemmas and non-trivial calculations. With our approach, finding the optimal state is equivalent to solving a 1D Harmonic oscillator. In fact for pure states the QFI equals 4 times the variance, hence we have 
\begin{equation}
  R_{\mathcal{F}}[\psi(x),x]=4\omega^{2}\left[\int_{0}^{\infty}x^{2}\abs{\psi(x)}^{2}-\left(\int_{0}^{\infty}x\abs{\psi(x)}^{2}\right)^{2}\right].
\end{equation}
If we relax the problem and allow the minimization over every regular $\psi$ defined over the entire real line, we can assume without loss of generality that $\int_{0}^{\infty}x|{\psi(x)}|^{2}=0$, making the QFI functional proportional the $\langle E^{2}\rangle$, hence the full action to minimize reads 
\begin{align} \label{eq:groundQFI}
\mathcal{S}_{\mathcal{F}}[\psi(x)] = \delta^{2} \int_{-\infty}^{\infty} \dd x\, \psi(x) \left( x^{2}a  - \frac{\dd^2}{\dd x^2}  \right) \psi^*(x).
\end{align}
It is then immediate that the optimal solutions are ground states of the harmonic oscillator i.e. $\psi_{\mathcal{F}}(x):=C_{\mathcal{F}}\exp\left[-\frac{2x^{2}\omega^{2}}{\mathcal{F}}\right]$ yielding a lower bound to the optimal unitary defect 
\begin{equation} \label{Eq:pureQFIIneq}
    {\rm UD}_{min}(\mathcal{F})\geq \frac{\omega^{2}}{\mathcal{F}}.
\end{equation}
As we show at the end of the section this \lq\lq unconstrained" optimum can be asymptotically reached just by shifting to the right the wave function of the ground state. 

However, since the QFI is convex in the state we cannot a priori rule out the existence of a mixed state with lower QFI and same unitary defect.
We now prove that such mixed states do not exist.
As it was proved for the first time in \cite{QFI-ConvexRoof}, for any state $\beta$ exist and ensemble $\{ \ket{\beta(i)}\}_{i}$ of pure states such that 
\begin{align}
    & \beta=\underset{i}{\sum}p_{i}\ketbra{\beta(i)}{\beta(i)}, \\
    & \mathcal{F}(\beta,H_{B})=\underset{i}{\sum}p_{i}\mathcal{F}(\ket{\beta(i)},H_{B}) .
\end{align}
Then we use that the functional $\epsilon_{C}(\mathbf{\Phi}_{\beta},\mathcal{V})$ is linear in $\beta$, hence we have 
\begin{align}
    &\frac{\epsilon_{C}(\mathbf{\Phi}_{\beta},\mathcal{V})}{\abs{V_{01}}^{2}}=\underset{i}{\sum}p_{i}\frac{\epsilon_{C}(\mathbf{\Phi}_{\ket{\beta(i)}},\mathcal{V})}{\abs{V_{01}}^{2}}= \underset{i}{\sum}\, p_{i}\,  {\rm UD}(\ket{\beta(i)})+o(\epsilon_{C}(\mathbf{\Phi}_{\beta},\mathcal{V})).
   \end{align}
   We use now Eq.\eqref{Eq:pureQFIIneq} on each component:
\begin{align}
     \underset{i}{\sum}p_{i} {\rm UD}(\ket{\beta(i)}) 
    \geq \underset{i}{\sum}p_{i}\frac{\omega^{2}}{\mathcal{F}(\ket{\beta(i)},H_{B})}
    \geq \frac{\omega^{2}}{\underset{i}{\sum}p_{i}\mathcal{F}(\ket{\beta(i)},H_{B})}=\frac{\omega^{2}}{\mathcal{F}(\beta,H_{B})}. 
\end{align}
Where the last inequality holds thanks to the convexity of $\frac{1}{x}$.\\
This result can be obtained from theorem 1 of \cite{QFIBound} using the relationship between worst-case infidelity and Choi infidelity, however, this section is meant to illustrate how our construction can be applied not only to concave resources or pure states, instead can find optimal bound over all states of convex resources too. 

We now prove that the inequality in Eq. \eqref{Eq:pureQFIIneq}
is infact an equality, chosing a specific family of functions $\psi_{\mathcal{F}}$  defined on $[0,\infty]$. Let 
\begin{equation} \label{Eq.QFI-optimalstate}
    \psi_{\mathcal{F}}(x):=C'_{\mathcal{F}} s(x) \exp\left[-\frac{2(x-\frac{\mathcal{F}}{\omega^{2}})^{2}\omega^{2}}{\mathcal{F}}\right],
\end{equation}
where $s(x)$ is a smooth function such that $s(0)=0$ and $\underset{x\rightarrow \infty}{\lim}x^{2}(s(x)-1)=0$, e.g., $s(x)=1-e^{-x^2}$ and  $C'_{\mathcal{F}}$ is the new normalization constant. It is immediate that in the $\mathcal{F} \rightarrow \infty$ limit the latter satisfies ${\rm UD}_{min}(\mathcal{F})= \frac{\omega^{2}}{\mathcal{F}}+o(\mathcal{F}^{-1})$, proving the thesis.
\subsection{Performances of coherent states}
 We now consider a semi-classical pulse, i.e. a battery prepared in a coherent state with parameter $\alpha$. Since the phase can always be reabsorbed into the interaction, consider coherent state of real $\alpha>0$. Asymptotically in the average number of photons $\langle n \rangle =\abs{\alpha}^{2} \rightarrow \infty$ the energy distribution of a coherent state is Gaussian with both average and variance equal to $\abs{\alpha}^{2}$, indeed we have
\begin{equation}
    \braket{n|\alpha}= e^{\frac{\alpha^{2}}{2}}\frac{\alpha^{n}}{\sqrt{n!}}= C_{\alpha} \psi_{\alpha}\left(n\right)+O(\alpha^{-3}),
\end{equation}
where $\psi_{\alpha}(x)=\exp[\frac{-(x-\alpha^{2})^{2}}{4\alpha^{2}}]$ and $C_{\alpha}=\sqrt{\sumab{n=0}{\infty}\abs{\psi_{\alpha}\left(n \right)}^{2}}$ is the normalization constant.
It is now apparent that the state in Eq.\eqref{Eq.QFI-optimalstate} is asymptotically a coherent state (except for the exponentially small modification around $0$ due to the regularizing $s(x)$ function).
Thus coherent states are optimal in the QFI regard, as for it's performances with respect to energy we have instead 
\begin{align}
    & {\rm UD}(\alpha):=C_{\alpha}^{2}\int_{0}^{\infty}\abs{\psi_{\alpha}'(x)}^{2}dx  =C_{\alpha}^{2}\int_{0}^{\infty}\frac{(x-\alpha^{2})}{4\alpha^{2}}\abs{\psi_{\alpha}(x)}^{2}dx =\frac{1}{4\alpha^{2}}+o(\alpha^{-2})=\frac{\omega}{4\langle E \rangle} +o\left(\frac{\omega}{\langle E \rangle}\right).
\end{align}
This reveals that the Choi infidelity scales as the inverse of the average energy, while we proved that the optimal scaling is the inverse of the square root of the average energy, see the first row of the Table I in the main text.

\section{Qudit target systems} \label{Sec:Qudit}
In previous sections, we have restricted our analysis to a qubit target system. Here we generalize the results for a target system of $d$ equally spaced energy levels. We consider a Hamiltonian of the target system of the form $ H_{S}=\omega \underset{n=0}{\overset{d-1}{\sum}} n\underset{S}{\ketbra{n}{n}} $ and a battery Hamiltonian
\begin{equation}\label{H_{B}^{(1)}}
    H_{B}:=\omega \underset{n=0}{\overset{\infty}{\sum}}n\underset{B}{\ketbra{n}{n}}.
\end{equation}
The choice above is quite natural, as we already discussed for the $d=2$ case. 
 The interaction is a straightforward generalization of the $d=2$ case:
\begin{equation} \label{Def:U_{SB}^{(1)}}
U_{SB}=I_{SB}^{d-1} +\underset{n=d-1}{\overset{\infty}{\bigoplus}} U^{(n)}.
\end{equation}
Where $I_{SB}^{d-1}:=I_{S}\otimes \sumab{j=0}{d-2}\underset{B}{\ketbra{j}{j}}$ is the identity on the first $d-1$ levels of $H_{B}$ and $U^{(n)}:=\underset{i,j=0}{\overset{d-1}{\sum}}\underset{B}{\ketbra{n-i}{n-j}} \otimes \underset{S}{\ketbra{i}{j}}V_{ij}$.
The first part of the calculations is formally identical to the case $d=2$, indeed 
we have for any $n\geq d-1$
\begin{align}
& V_{S}^{\dagger}K^{(n)}=  \beta_{n}I_{S}+J_S^{(n)},
\end{align}
where $J_S^{(n)}:=\underset{i,j,j'=0}{\overset{d-1}{\sum}}\underset{S}{\ketbra{i}{j}} (\beta_{n+i-j'}-\beta_{n}) V_{ij'} V_{jj'}^{*}$. At this stage, we assume that the coefficient are smooth, i.e. that exists a regular function $\psi$ such that $\beta_{n}^{(\delta)}:=C_{\delta}\psi(n\delta)$ and are interested in the $\delta \rightarrow 0$ limit. As in the $d=2$ case populating the first $d-1$ level of the battery is strongly sub-optmal, so we assume that $\psi(0)=0$, implying $\abs{\beta_{0}}^{2},..\abs{\beta_{d-1}}^{2}\propto \delta^{3}$, and making their contribution to the Choi infidelity negligible. Then we have  
\begin{equation} \label{eq:Choi-InfIntermediatd}
    \epsilon_{C}(\mathbf{\Phi}_{\beta^{(\delta)}},\mathcal{V})=-\frac{2}{d}\underset{n>d-2}{\sum} \text{Re}\left[C_{\delta}\psi(\delta n)^{*}\mbox{Tr}[J_S^{(n)}]\right]-\underset{n>d-2}{\sum} \frac{1}{d^{2}} \abs{\mbox{Tr}[J_S^{(n)}]}^{2}+O(\delta^{3}),
\end{equation}
the trace of $J_S^{(n)}$ reads 
\begin{equation} \label{Eq:Trace-J-d}
    \mathrm{Tr}[J_S^{(n)}] = \sum_{j=0}^{d-1} \left[ (\beta_{n+j}-\beta_{n}) \sum_{\substack{l-k=j}} \lvert V_{l,k} \rvert^{2} - (\beta_{n}-\beta_{n-j}) \sum_{\substack{l-k=j}} \lvert V_{k,l} \rvert^{2} \right].   
\end{equation}
 
Hence, we can make the substitutions
\begin{align}
& \beta^{(\delta)}_{n+j}-\beta^{(\delta)}_{n}=C_{\delta}\left( j \delta \psi'(n\delta)+\frac{1}{2}\delta^{2} j^{2} \psi''(n\delta) +O(\delta^{3})\right),  \\
& \beta^{(\delta)}_{n}-\beta^{(\delta)}_{n-j}=C_{\delta}\left( j \delta \psi'(n\delta)- \frac{1}{2}\delta^{2} j^{2} \psi''(n\delta)+O(\delta^{3})\right),
\end{align}
and plug them in Eq.\eqref{Eq:Trace-J-d}, getting 
\begin{equation} \label{AsymptoticChoiFid (d level)}
    \epsilon_{C}(\mathbf{\Phi},\mathcal{V})=-\frac{2}{d}\abs{C_{\delta}}^{2}\text{Re}\underset{n>0}{\overset{N-1}{\sum}}\sum_{j=0}^{d-1} \psi^{*}(n\delta) \frac{j^{2}}{2}\psi''(n\delta)\left[  \sum_{\substack{l-k=j}} \lvert V_{l,k} \rvert^{2} +  \sum_{\substack{l-k=j}} \lvert V_{k,l} \rvert^{2} +O(\delta^{3})\right] .
\end{equation}
where the terms proportional to $\psi'(n\delta)$ disappeared thanks to a summation by parts and vanishing boundary conditions. Similarly one can show that $\underset{n>0}{\sum}\abs{\mbox{Tr}[J_S^{(n)}]}^{2}=O(\delta^{4})$ and thus neglect this contribution.
We conclude that for smooth battery states the Choi infidelity is once again proportional to the unitary defect: 
\begin{equation} \label{Eq:GeneralizedChoiComp}
    \epsilon_{C}(\mathbf{\Phi}_{\beta^{\delta}},\mathcal{V})=\mathcal{A}[V_{S}] \delta^{2}\,  {\rm UD}(\psi)+O\left( \delta^{3} \right).
\end{equation}
Where $\mathcal{\mathcal{A}}[V_{S}]:=\underset{j=1}{\overset{d}{\sum}}\frac{j^{2}}{d}\underset{\abs{l-k}=j}{\sum} \abs{V_{l,k}} ^{2} $ is a very intuitive quantifier of asymmetry that for $d=2$ reduces to $\abs{V_{01}}^{2}$ as expected.

\subsection{An alterative scheme for gates with invariant subspaces block diagonal in energy.}
In this sub-section, we highlight that for $d>2$ the above construction is not always optimal, as there are target gates $V_{S}$ for which more efficient battery state and interaction exist. As an example, consider a gate that acts non-trivially only on the ground and maximally exited states, i.e. on  $\mathcal{H}_{2}:=\text{Span}\{\ket{0}_{S};\ket{d-1}_{S}\}$. Then  $\mathcal{H}_{S}=\mathcal{H}_{2}\oplus \mathcal{H}_{\perp} $, and $V_{S}:= V_{2}\oplus I_{\perp}$, where $V_{2}$ is a unitary acting on the subspace $\mathcal{H}_{2}$.
We now compare the performances of the previous scheme of implementation, with an alternative implementation that ignores all the levels but the extremal ones, and involves a battery matching their energy gap, essentially reducing the task to the qubit problem.

The first scheme (I) consist of using a battery Hamiltonian $H_{B}^{(\rm{I})}$ given in Eq.(\ref{H_{B}^{(1)}}) with $H_{B}^{(\rm{I})}\ket{n^{(\rm{I})}}_{B}=\omega n \ket{n^{(\rm{I})}}_{B} $ , the interaction $U_{SB}^{(\rm{I})}$ in Eq.(\ref{Def:U_{SB}^{(1)}}) and the battery state 
\begin{equation}
    \ket{\beta^{(\rm{I})}_{\delta}}_{B}=C_{\delta}\sum_{n=0}^{\infty} \psi\left(n\delta \right)\ket{n^{(\rm{I})}}_{B}.
\end{equation}
where $\psi$ is smooth and vanishing in $0,\infty$ as in the main text. The channel of scheme (I) is thus defined as 
\begin{equation}
    \mathbf{\Phi}_{\delta}^{(\rm{I})}(\cdot):=Tr_{B}[\mathcal{U}_{SB}^{(\rm{I})}( \cdot \otimes \beta_{\delta}^{(\rm{I})})].
\end{equation}
 It's Choi infidelity can be computed trough Eq.\eqref{Eq:GeneralizedChoiComp} and reads  
\begin{equation}
    \epsilon_{C}(\mathbf{\Phi}^{(\rm{I})}_{\delta},\mathcal{V}_{S})=\mathcal{A}[V_{S}]\delta^{2} {\rm UD}(\psi)+O(\delta^{3})=\frac{2(d-1)^{2}}{d}\delta^{2} {\rm UD}(\psi)\abs{\bra{0}V_{S}\ket{d-1}}^{2}+O(\delta^{3}).
\end{equation}

The second scheme (II) is conceptually identical, except for the gap energy of the battery Hamiltonian, specifically with
 \begin{equation}
    H_{B}^{(\rm{II})}=\omega (d-1) \underset{n=0}{\overset{N-1}{\sum}}n\underset{B}{\ketbra{n^{(\rm{II})}}{n^{(\rm{II})}}},
\end{equation}
and the interaction is defined as
\begin{align}
   &U_{SB}^{(\rm{II})}=\underset{B}{\ketbra{0}{0}}\otimes \underset{S}{\ketbra{0}{0}} + \underset{n>0}{\bigoplus}U^{(n)} + I_{B}\otimes \Pi_{\perp}.
\end{align}
Where $U^{(n)}$ act on the two dimensional degenerate subspace of energy $(d-1)n \omega $, precisely 
\begin{equation}
    U^{(n)}:=\suma{ij=0,1}\underset{B}{\ketbra{n-i}{n-j}}\otimes \underset{S}{\ketbra{(d-1)i}{(d-1)j}}\bra{i}V_{2}\ket{j}
\end{equation}
and $\Pi_{\perp}$ is the projector onto the $\mathcal{H}_{\perp}$ subspace.

Finally, we choose for a fair comparison the battery state $\beta^{(\rm{II})}$ to have the same ``shape":
\begin{equation}
    \ket{\beta^{(\rm{II})}_{\delta}}_{B}=C_{\delta}\sum_{n=0}^{\infty} \psi\left(\delta n\right)\ket{n^{(\rm{II})}}_{B}.
\end{equation}
As expected, the Kraus operators of scheme (II) are completely identical to the qubit case, using $U_{SB}^{(\rm{II})} = U_{SB}^{\{0,d-1\}}\oplus (\Pi_\perp\otimes I_B)$ where $\Pi_\perp =\sum_{i=2}^{d-1} \ketbra{i}_S$,
\begin{align}
   & V_{S}^{\dagger}K_{n}^{(\rm{II})} = V_S^\dag \left(\bra{n} U_{SB}^{\{0,d-1\}} \ket{\beta}  + \bra{n} \Pi_\perp\otimes I_B \ket{\beta} \right) = \big(V^{(\rm{II})\dag} \bra{n} U_{SB}^{\{0,d-1\}} \ket{\beta}\Big)  \oplus \beta_n \Pi_\perp
   =\psi(n\delta)I_{S}+J_{n}^{(\rm{II})}
\end{align}
where the \lq\lq error" operators are given by 
\begin{align}
   &   J_{n}^{(\rm{II})}:=\underset{ijj'=0,1}{\sum}\underset{S}{\ketbra{i(d-1)}{j(d-1)}} [\psi\left((n+i-j')\delta\right)-\psi(n\delta)] V_{ij'}^{(\rm{II})} V_{jj'}^{*(\rm{II})}.
\end{align}
Let $\Delta_{n}:=tr[J_{n}^{(\rm{II})}]$, then recalling that $\abs{\Delta_{n}}^{2}\propto \delta^{4}$ we have 
\begin{align}
    & \epsilon_{C}(\mathbf{\Phi}^{(\rm{II})}_{\delta},\mathcal{V}_{S})=1-\frac{1}{d^{2}}\suma{n>0}\abs{ tr[V_{S}^{\dagger}K_{n}^{(\rm{II})}]}^{2}= 1- \frac{1}{d^{2}}\suma{n>0}(\abs{\psi(n\delta)}^{2}tr[I_{S}]^{2}+2tr[I_{S}]\text{Re}[\psi(n\delta)\Delta_{n}])+O(\delta^{4})= \\
    & -\suma{n>0} \frac{2}{d}\text{Re}[\psi(n\delta)\Delta_{n}]+O(\delta^{4}).
\end{align}
The rest of the calculation is done by comparison with the qubit case, the only difference is this $\frac{2}{d}$ additional factor.

Thus the asymptotic Choi infidelity can essentially be found thanks to Eq.(11) of the main text, 
\begin{align}
    \epsilon_{C}(\mathbf{\Phi}^{(\rm{II})}_{\delta},\mathcal{V}_{S})=\frac{2}{d}\delta^{2}{\rm UD}(\psi)\abs{\bra{0}V_{S}\ket{d-1}}^{2}+O(\delta^{3})
\end{align}

To complete the comparison we list the resources possessed by the two battery states, evaluated at equal Choi infidelity, i.e. fixing $\delta^{(\rm{II})}=(d-1)\delta^{(\rm{I})}$. For energy and QFI there is no advantage in scheme (\rm{II}), as the boost in Choi Infidelity is compensated by the larger gap energies of the Hamiltonian. However if $\psi(x)$ has compact support the number of levels used is reduced by a factor $d-1$. More in general, for any smooth function $\psi$ there is a net saving of coherence as quantified by the entropic coherence \cite{EntropicCoherence(2025)} 
\begin{align}
& C(\beta^{(\rm{II})}_{(d-1)\delta},H_{B}^{(\rm{II})})=C(\beta^{(\rm{I})}_{\delta},H_{B}^{(\rm{I})})-\log(d-1)+o(1).
\end{align}
 In conclusion, for $d>2$ scheme (II) uses less coherence at fixed precision, while being equally efficient in terms of (dimensional) energetic resources. 
 This example shows a general important point: if $V_{S}$ has non trivial (i.e. non one- dimensional) invariant subspaces that are block diagonal in the energy eigen-base, alternative and potentially more efficient schemes of implementation exist. Thus proving the optimality of a certain interaction and Hamiltonian in the $d>2$ case is considerably harder. This is especially important in the case of a system $S$ made of multiple qubits, the total Hamiltonian $H_{S}$ has large degenerate eigenspaces.


\end{widetext}

\end{document}